\pdfoutput=1
\documentclass[11pt,a4paper]{article}%
\usepackage{amsmath}
\usepackage{amsfonts}
\usepackage{amssymb}
\usepackage{graphicx}
\usepackage[onehalfspacing]{setspace}
\providecommand{\U}[1]{\protect\rule{.1in}{.1in}}
\newtheorem{theorem}{Theorem}

\newtheorem{algorithm}[theorem]{Algorithm}

\newtheorem{proposition}[theorem]{Proposition}

\newenvironment{proof}[1][Proof]{\noindent\textbf{#1.} }{\ \rule{0.5em}{0.5em}}
\begin{document}

\title{\textbf{Wald type and Phi-divergence based test-statistics\ for isotonic
binomial proportions}}
\author{Martin, N.$^{1}$, Mata, R.$^{2}$ and Pardo, L.$^{2}$\\$^{1}${\small Dep. Statistics, Carlos III University of Madrid, 28903 Getafe
(Madrid), Spain}\\$^{2}${\small Dep. Statistics and O.R., Complutense University of Madrid,
28040 Madrid, Spain} }
\date{\today}
\maketitle

\begin{abstract}
In this paper new test statistics are introduced and studied for the important
problem of testing hypothesis that involves inequality constraint on
proportions when the sample comes from independent binomial random variables:
Wald type and phi-divergence based test-statistics. As a particular case of
phi-divergence based test-statistics, the classical likelihood ratio test is
considered. An illustrative example is given and the performance of all of
them for small and moderate sample sizes is analyzed in an extensive
simulation study.

\end{abstract}

\bigskip


\section{Introduction\label{Sec1}}

Ordinal categorical data appear frequently in the biomedical research
literature, for example, in the analysis of $I$ independent binary random
variables related to an increasing ordered categorical variable. It is
well-known that for such data it is not possible to use the classical
test-statistics such as chi-square or likelihood ratio with chi-squared
asymptotic distribution, but there exist appropriate order-restricted
test-statistics with chi-squared-bar asymptotic distribution. To illustrate
this problem a modification of an example given in Silvapulle and Sen (2005)
is considered in this introductory section.%

\begin{table}[htbp]  \tabcolsep2.8pt  \centering
\begin{tabular}
[c]{||clc||crr||crc|crc||}\hline
\multicolumn{3}{||c||}{$i$} & \multicolumn{3}{c||}{$n_{i}$} &
\multicolumn{3}{c|}{$n_{i1}$} & \multicolumn{3}{|c||}{$n_{i2}$}\\\hline
\hspace*{1cm} & $1$ & \hspace*{1cm} & \hspace*{1cm} & $17114$ & \hspace*{1cm}
& \hspace*{1cm} & $48$ & \hspace*{1cm} & \hspace*{1cm} & $17066$ &
\hspace*{1cm}\\
& $2$ &  &  & $14502$ &  &  & $38$ &  &  & $14464$ & \\
& $3$ &  &  & $793$ &  &  & $5$ &  &  & $788$ & \\
& $4$ &  &  & $165$ &  &  & $2$ &  &  & $163$ & \\\hline
\end{tabular}
$\ \ \ \ \ \ \ \ \ \ \ \ $%
\caption{Number of individuals with (j=1) and without (j=0) congenital sex-organ malformation cross-classified according to the maternal alcohol consumption level (i=1,2,3,4).\label{tt1}}%
\end{table}%

Table \ref{tt1} contains a subset of data from a prospective study of maternal
drinking and congenital malformations. Women completed a questionnaire, early
in their pregnancy, concerning alcohol use in the first trimester; complete
data and details are available elsewhere (Graubard and Korn, 1987).
Specifically, women were asked what was the amount of alcohol taken during the
first three months of their pregnancy and four categories of drink doses are
considered ($I=4$), no alcohol consumption ($i=1$), average number of
alcoholic drinks per day less than one but greater than zero ($i=2$),\ one or
more and less than three alcoholic drinks per day ($i=3$) and three or more
alcoholic drinks per day ($i=4$). In terms of a binary random variable with
$n_{i}$ individuals in total (see the second column in Table \ref{tt1}) with
independent behavior with respect to having congenital malformations, the
individuals not having congenital malformations are considered to be
unsuccessful ($j=2$, see the last column in Table \ref{tt1}) and successful
otherwise ($j=1$, see the third column in Table \ref{tt1}). Let $\pi_{i}$ be
the probability of a success associated with the $i$-th alcohol dose. Let us
consider some statistical inference questions that may arise in this example
and in similar ones with binomial probabilities.

\begin{enumerate}
\item Is there any evidence of maternal alcohol consumption being related to
malformation of sex organ? To answer this question, the null and alternative
hypotheses may be formulated as
\[
H_{0}:\pi_{1}=\pi_{2}=\pi_{3}=\pi_{4}\text{ vs. }H_{1}:\pi_{1},\pi_{2},\pi
_{3},\pi_{4}\text{ are not all equal,}%
\]
respectively. However, this formulation is unlikely to be appropriate because
the main issue of interest is the possible increase in the probability of
malformation as alcohol consumption increases.

\item Is there any evidence that an increase in maternal alcohol consumption
is associated with an increase in the probability of malformation?. This
question, as it stands, is quite broad to give a precise formulation of the
null and the alternative hypotheses. One possibility is to formulate the
problem in the following way,%
\begin{equation}
H_{0}:\pi_{1}=\pi_{2}=\pi_{3}=\pi_{4}\text{ vs. }H_{1}:\pi_{1}\leq\pi_{2}%
\leq\pi_{3}\leq\pi_{4}\text{ with at least one inequality being strict.}
\label{J1}%
\end{equation}

\end{enumerate}

Consider an experiment with $I$ increasing ordinal categories for a variable
$X$. Suppose that $n_{i}$ prefixed individuals are assigned to the $i$-th
category and $n=\sum_{i=1}^{I}n_{i}$. The individuals are followed over time
for the development of an event of interest $Y$ and the events related to the
individuals are independent. Let $N_{i1}$ be the random variable that
represents the number of individuals related to successful events ($Y=1$) out
of the total assigned to the $i$-th category, $n_{i}$, $i=1,...,I.$ If we
denote by $\pi_{i}=\Pr(Y=1|X=i)$ the probability of a success associated with
the $i$-th category, we have that $N_{i1}$ is a Binomial random variable with
parameters $n_{i}$ and $\pi_{i}$, $i=1,...,I$. Let $N_{i2}$ denote the number
of unsuccessful events associated with the $i$-th category, i.e. $N_{i2}%
=n_{i}-N_{i1}$, then the contingency table of a realization of $(N_{i1}%
,N_{i2})$, $i=1,...,I$, is in the last two columns of the following table%
\[%
\begin{tabular}
[c]{||c||cc||}\hline\hline
$n_{1}$ & $n_{11}$ & $n_{12}=n_{1}-n_{11}$\\
$\vdots$ & $\vdots$ & $\vdots$\\
$n_{i}$ & $n_{i1}$ & $n_{i2}=n_{i}-n_{i1}$\\
$\vdots$ & $\vdots$ & $\vdots$\\
$n_{I}$ & $n_{I1}$ & $n_{I2}=n_{I}-n_{I1}$\\\hline\hline
\end{tabular}
\ .
\]

Our purpose in this paper is to propose new order-restricted test statistics,
Wald-type and phi-divergence based test-statistics for testing%
\begin{align}
H_{0}  &  :\pi_{1}=\pi_{2}=\cdots=\pi_{I},\label{J1a}\\
H_{1}  &  :\pi_{1}\leq\pi_{2}\leq\cdots\leq\pi_{I}\text{ with at least one
inequality being strict.}\nonumber
\end{align}
The classical likelihood ratio test statistic will appear as a particular case
of phi-divergence based test-statistics. A log-linear formulation of
(\ref{J1}) is proposed in Section \ref{Sec1b}, fundamental for defining the
Wald type test-statistics. In Section \ref{Sec2} the families of
phi-divergence test statistics are presented. Section \ref{Sec3} is devoted to
solve the problem presented in this Section \ref{Sec1} for a illustrative
example. An extensive simulation study is carried out in Section \ref{Sec4}.

\section{Formulation for isotonic binomial proportions in terms of log-linear
and logistic regression modeling: Wald type test-statistics\label{Sec1b}}

Reparametrizating the initial problem through log-linear modeling, the
formulation of the null hypothesis is strongly simplified since all the
interaction parameters are zero under the null hypothesis and this is
appealing, in special, to create Wald type test-statistics. Let%
\begin{align}
\boldsymbol{p}  &  =\boldsymbol{p}(\boldsymbol{\theta})=(p_{11}%
(\boldsymbol{\theta}),p_{12}(\boldsymbol{\theta}),p_{21}(\boldsymbol{\theta
}),p_{22}(\boldsymbol{\theta}),...,p_{I1}(\boldsymbol{\theta}),p_{I2}%
(\boldsymbol{\theta}))^{T}\nonumber\\
&  =(\tfrac{n_{1}}{n}\pi_{1},\tfrac{n_{1}}{n}(1-\pi_{1}),\tfrac{n_{2}}{n}%
\pi_{2},\tfrac{n_{2}}{n}(1-\pi_{2}),...,\tfrac{n_{I}}{n}\pi_{I},\tfrac{n_{I}%
}{n}(1-\pi_{I}))^{T} \label{prob}%
\end{align}
be the probability vector of the following saturated log-linear model%
\begin{equation}
\log p_{ij}(\boldsymbol{\theta})=u+u_{1(i)}+\theta_{2(j)}+\theta_{12(ij)},
\label{loglineq}%
\end{equation}
with%
\begin{equation}
u_{1(I)}=0,\quad\theta_{2(2)}=0,\quad\theta_{12(i2)}=0,i=1,...,I-1,\quad
\theta_{12(Ij)}=0,j=1,2, \label{ident}%
\end{equation}
being the identifiability constraints,%
\begin{equation}
\boldsymbol{\theta}=(\theta_{2(1)},\theta_{12(11)},...,\theta_{12(1I-1)})^{T}
\label{theta}%
\end{equation}
the unknown parameters vector and $\boldsymbol{u}=(u,u_{1(1)},...,u_{1(I-1)}%
)^{T}$ with%
\begin{equation}
u=u(\boldsymbol{\theta})=\log\frac{n_{I}/n}{1+\exp\{\theta_{2(1)}\}},
\label{u}%
\end{equation}%
\begin{equation}
u_{1(i)}=u_{1(i)}(\boldsymbol{\theta})=\log\frac{\frac{n_{i}}{n_{I}}\left(
1+\exp\{\theta_{2(1)}\}\right)  }{1+\exp\{\theta_{2(1)}+\theta_{12(i1)}%
\}},\quad i=1,...,I-1, \label{theta1}%
\end{equation}
the redundant parameters, obtained through $\boldsymbol{\theta}$ taking into
account $p_{i1}(\boldsymbol{\theta})+p_{i0}(\boldsymbol{\theta})=\frac{n_{i}%
}{n}$, $i=1,...,I$. In terms of the log-linear formulation, (\ref{J1a}) is
equivalent to%
\begin{align}
H_{0}  &  :\theta_{12(11)}=\theta_{12(21)}=\cdots=\theta_{12(I-1,1)}%
=0,\label{J1b}\\
H_{1}  &  :\theta_{12(11)}\leq\theta_{12(21)}\leq\cdots\leq\theta
_{12(I-1,1)}\leq0\text{ with at least one inequality being strict.}\nonumber
\end{align}
Notice that $\theta_{2(1)}$ is a nuisance parameter since it does not
interfere in (\ref{J1b}). In particular, under the null hypothesis of $\pi
_{0}=\pi_{1}=\pi_{2}=\cdots=\pi_{I}$, the value of the nuisance parameter is
$\theta_{2(1)}=\mathrm{logit}(\pi_{0})=\log[\pi_{0}/(1-\pi_{0})]$, and thus it
contains all the information about the homogeneous probability vector.

In matrix notation, we can express the vector of parameters of the log-linear
model in terms of the following logistic regression%
\begin{equation}
\mathrm{logit}(\boldsymbol{\pi})=\boldsymbol{X\theta} \label{logistic}%
\end{equation}
where
\[
\boldsymbol{\pi}=(\pi_{1},\pi_{2},...,\pi_{I})^{T},\qquad\boldsymbol{X}=%
\begin{pmatrix}
\boldsymbol{1}_{I-1} & \boldsymbol{I}_{I-1}\\
1 & \boldsymbol{0}_{I-1}^{T}%
\end{pmatrix}
,
\]
$\boldsymbol{I}_{a}$\ is the the identity matrix of order $a$, $\boldsymbol{1}%
_{a}$\ is the $a$-vector of ones and $\boldsymbol{0}_{a}$\ is the $a$-vector
of zeros. Since a saturated model has been considered, $\boldsymbol{X}$ is a
full rank matrix and thus we can consider%
\begin{equation}
\boldsymbol{\theta}=\boldsymbol{X}^{-1}\mathrm{logit}(\boldsymbol{\pi}),
\label{loglin}%
\end{equation}
and on the other hand (\ref{J1b}) in matrix notation is given by%
\begin{align}
H_{0}  &  :\boldsymbol{R\theta}=\boldsymbol{0}_{I-1},\label{J1c}\\
H_{1}  &  :\boldsymbol{R\theta}\leq\boldsymbol{0}_{I-1}\text{ and
}\boldsymbol{R\theta}\neq\boldsymbol{0}_{I-1}\text{,}\nonumber
\end{align}
with $\boldsymbol{R}=(\boldsymbol{0}_{I-1},\boldsymbol{G}_{I-1})$, and
$\boldsymbol{G}_{I-1}$ is a square matrix of order $I-1$ with $1$-s in the
main diagonal and $-1$-s in the upper superdiagonal.

We shall consider three parameter spaces for $\boldsymbol{\theta}$%
\[
\Theta_{0}=\left\{  \boldsymbol{\theta}\in%
\mathbb{R}
^{I}:\boldsymbol{R\theta}=\boldsymbol{0}_{I-1}\right\}  \subset
\widetilde{\Theta}=\left\{  \boldsymbol{\theta}\in%
\mathbb{R}
^{I}:\boldsymbol{R\theta}\leq\boldsymbol{0}_{I-1}\right\}  \subset\Theta=%
\mathbb{R}
^{I},
\]
i.e. while $\Theta$ is the restricted parameter space, $\bar{\Theta}$ is
unrestricted, and $\Theta_{0}$ becomes the parameter space under the null
hypothesis. It is well known that, the Fisher information matrix for
$\boldsymbol{\theta}\in\Theta$ in the logistic regression is given by
\begin{equation}
\mathcal{I}_{F}(\boldsymbol{\theta})=\boldsymbol{X}^{T}\mathrm{diag}\{\nu
_{i}\pi_{i}(1-\pi_{i})\}_{i=1}^{I}\boldsymbol{X}, \label{FIM1}%
\end{equation}
where $\nu_{i}=\lim_{n\rightarrow\infty}\frac{n_{i}}{n}$, $i=1,...,I$. The
following result provides the explicit expression of the Fisher information
matrix under the null hypothesis given in (\ref{J1a}) or (\ref{J1b}).

\begin{theorem}
\label{th0}For $\boldsymbol{\theta}_{0}\in\Theta_{0}$ in the model
(\ref{loglineq}) or (\ref{logistic}), the Fisher information matrix is given
by%
\begin{equation}
\mathcal{I}_{F}(\widehat{\boldsymbol{\theta}})=\pi_{0}(1-\pi_{0})%
\begin{pmatrix}
1 & \nu_{1} & \nu_{2} & \cdots & \nu_{I-1}\\
\nu_{1} & \nu_{1} & 0 & \cdots & 0\\
\nu_{2} & 0 & \nu_{2} &  & 0\\
\vdots & \vdots &  & \ddots & \vdots\\
\nu_{I-1} & 0 & 0 & \cdots & \nu_{I-1}%
\end{pmatrix}
. \label{FIM2}%
\end{equation}

\end{theorem}

\begin{proof}
It is immediate by plugging $\pi_{i}=\pi_{0}$, $i=1,...,I$, to (\ref{FIM1}).
\end{proof}

If $\widehat{\boldsymbol{\theta}}$, $\widetilde{\boldsymbol{\theta}}$ and
$\overline{\boldsymbol{\theta}}$ represent the maximum likelihood estimator
(MLE) of $\boldsymbol{\theta}$\ focussed on the parameter spaces, $\Theta_{0}%
$, $\widetilde{\Theta}$, $\Theta$ respectively, according to Silvapulle and
Sen (2005, pages 154 and 166) we can consider three Wald-type test-statistics,%
\begin{align}
W(\widetilde{\boldsymbol{\theta}},\widehat{\boldsymbol{\theta}})  &
=n\widetilde{\boldsymbol{\theta}}^{T}\boldsymbol{R}^{T}\left(  \boldsymbol{R}%
\mathcal{I}_{F}^{-1}(\widehat{\boldsymbol{\theta}})\boldsymbol{R}^{T}\right)
^{-1}\boldsymbol{R}\widetilde{\boldsymbol{\theta}},\label{w0}\\
H(\widetilde{\boldsymbol{\theta}},\widehat{\boldsymbol{\theta}})  &
=n(\widetilde{\boldsymbol{\theta}}-\widehat{\boldsymbol{\theta}}%
)^{T}\mathcal{I}_{F}(\widehat{\boldsymbol{\theta}}%
)(\widetilde{\boldsymbol{\theta}}-\widehat{\boldsymbol{\theta}}),\label{w1}\\
D(\overline{\boldsymbol{\theta}},\widetilde{\boldsymbol{\theta}}%
,\widehat{\boldsymbol{\theta}})  &  =n(\overline{\boldsymbol{\theta}%
}-\widehat{\boldsymbol{\theta}})^{T}\mathcal{I}_{F}%
(\widehat{\boldsymbol{\theta}})(\overline{\boldsymbol{\theta}}%
-\widehat{\boldsymbol{\theta}})-n(\overline{\boldsymbol{\theta}}%
-\widetilde{\boldsymbol{\theta}})^{T}\mathcal{I}_{F}%
(\widetilde{\boldsymbol{\theta}})(\overline{\boldsymbol{\theta}}%
-\widetilde{\boldsymbol{\theta}}), \label{w2}%
\end{align}
where%
\begin{align*}
\mathcal{I}_{F}(\widehat{\boldsymbol{\theta}})  &  =\widehat{\pi}%
_{0}(1-\widehat{\pi}_{0})%
\begin{pmatrix}
1 & (\widehat{\boldsymbol{\nu}}^{\ast})^{T}\\
\widehat{\boldsymbol{\nu}}^{\ast} & \mathrm{diag}(\widehat{\boldsymbol{\nu}%
}^{\ast})
\end{pmatrix}
,\\
\mathcal{I}_{F}(\widetilde{\boldsymbol{\theta}})  &  =\boldsymbol{X}%
^{T}\mathrm{diag}\{\widehat{\nu}_{i}\widetilde{\pi}_{i}(1-\widetilde{\pi}%
_{i})\}_{i=1}^{I}\boldsymbol{X},\\
\widehat{\boldsymbol{\nu}}^{\ast}  &  =(\widehat{\nu}_{1},...,\widehat{\nu
}_{I-1})^{T}=(\frac{n_{1}}{n},...,\frac{n_{I-1}}{n})^{T}.
\end{align*}
These test-statistics, have according to Proposition 4.4.1 in Silvapulle and
Sen (2005), the same asymptotic distribution as the likelihood ratio test-statistic.

\begin{proposition}
Under the null hypothesis given in (\ref{J1a}) or (\ref{J1b}), the expression
of $W(\widetilde{\boldsymbol{\theta}},\widehat{\boldsymbol{\theta}})$,\ given
in (\ref{w0}), is as follows%
\[
W(\widetilde{\boldsymbol{\theta}},\widehat{\boldsymbol{\theta}})=n\widehat{\pi
}_{0}(1-\widehat{\pi}_{0})(\widetilde{\boldsymbol{\theta}}^{\ast}%
)^{T}\boldsymbol{\Sigma}_{\widehat{\boldsymbol{\nu}}^{\ast}}%
\widetilde{\boldsymbol{\theta}}^{\ast},
\]
where $\boldsymbol{\theta}^{\ast}=(\theta_{12(11)},...,\theta_{12(1I-1)})^{T}%
$, $\boldsymbol{\Sigma}_{\widehat{\boldsymbol{\nu}}^{\ast}}=\mathrm{diag}%
(\widehat{\boldsymbol{\nu}}^{\ast})-\widehat{\boldsymbol{\nu}}^{\ast
}(\widehat{\boldsymbol{\nu}}^{\ast})^{T}$,%
\begin{align*}
\widehat{\pi}_{0}  &  =\tfrac{N_{\bullet1}}{n},\\
N_{\bullet1}  &  =N_{11}+N_{21}+...+N_{I1}.
\end{align*}

\end{proposition}

\begin{proof}
From Theorem \ref{th0} the $2\times2$ block structure of the Fisher
information matrix is%
\[
\mathcal{I}_{F}(\boldsymbol{\theta}_{0})=\widehat{\pi}_{0}(1-\widehat{\pi}%
_{0})%
\begin{pmatrix}
1 & (\widehat{\boldsymbol{\nu}}^{\ast})^{T}\\
\widehat{\boldsymbol{\nu}}^{\ast} & \mathrm{diag}(\widehat{\boldsymbol{\nu}%
}^{\ast})
\end{pmatrix}
\]
and%
\[
\mathcal{I}_{F}^{-1}(\boldsymbol{\theta}_{0})=\frac{1}{\widehat{\pi}%
_{0}(1-\widehat{\pi}_{0})}%
\begin{pmatrix}
a & \boldsymbol{b}^{T}\\
\boldsymbol{c} & \boldsymbol{D}%
\end{pmatrix}
,
\]
with $\boldsymbol{D}=\boldsymbol{\Sigma}_{\widehat{\boldsymbol{\nu}}^{\ast}%
}^{-1}$, $\boldsymbol{\Sigma}_{\widehat{\boldsymbol{\nu}}^{\ast}%
}=\mathrm{diag}(\widehat{\boldsymbol{\nu}}^{\ast})-\widehat{\boldsymbol{\nu}%
}^{\ast}(\widehat{\boldsymbol{\nu}}^{\ast})^{T}$. But%
\[
\boldsymbol{R}\mathcal{I}_{F}^{-1}(\widehat{\boldsymbol{\theta}}%
)\boldsymbol{R}^{T}=\frac{1}{\widehat{\pi}_{0}(1-\widehat{\pi}_{0}%
)}\boldsymbol{G}_{I-1}\boldsymbol{\Sigma}_{\widehat{\boldsymbol{\nu}}^{\ast}%
}^{-1}\boldsymbol{G}_{I-1}^{T}%
\]
and%
\[
\boldsymbol{R}\widetilde{\boldsymbol{\theta}}=\boldsymbol{G}_{I-1}%
\boldsymbol{\theta}^{\ast}.
\]
Therefore,%
\begin{align*}
W(\widetilde{\boldsymbol{\theta}},\widehat{\boldsymbol{\theta}})  &
=n\widetilde{\boldsymbol{\theta}}^{T}\boldsymbol{R}^{T}\left(  \boldsymbol{R}%
\mathcal{I}_{F}^{-1}(\widehat{\boldsymbol{\theta}})\boldsymbol{R}^{T}\right)
^{-1}\boldsymbol{R}\widetilde{\boldsymbol{\theta}}\\
&  =n\widehat{\pi}_{0}(1-\widehat{\pi}_{0})(\widetilde{\boldsymbol{\theta}%
}^{\ast})^{T}\boldsymbol{\Sigma}_{\widehat{\boldsymbol{\nu}}^{\ast}%
}\widetilde{\boldsymbol{\theta}}^{\ast}.
\end{align*}

\end{proof}

There is an explicit formula for the MLEs of $\boldsymbol{\theta}$ under the
null hypothesis,%
\begin{equation}
\widehat{\boldsymbol{\theta}}=(\widehat{\theta}_{2(1)},\widehat{\theta
}_{12(11)},...,\widehat{\theta}_{12(1I-1)})^{T}=(\mathrm{logit}(\widehat{\pi
}_{0}),0,...,0)^{T}. \label{thetaHat}%
\end{equation}
For the calculation of $\widetilde{\boldsymbol{\theta}}=(\widetilde{\theta
}_{2(1)},\widetilde{\theta}_{12(11)},...,\widetilde{\theta}_{12(1I-1)})^{T}$
or $\overline{\boldsymbol{\theta}}=(\overline{\theta}_{2(1)},\overline{\theta
}_{12(11)},...,\overline{\theta}_{12(1I-1)})^{T}$, it is much easier to
calculate first the corresponding MLE for the probability vector,
$\widetilde{\boldsymbol{\pi}}=(\widetilde{\pi}_{1},\widetilde{\pi}%
_{2},...,\widetilde{\pi}_{I})^{T}$ or $\overline{\boldsymbol{\pi}}%
=(\overline{\pi}_{1},\overline{\pi}_{2},...,\overline{\pi}_{I})^{T}$, and
plugging it to (\ref{loglin}). There is an explicit formula for%
\begin{equation}
\overline{\boldsymbol{\pi}}=(\tfrac{N_{11}}{n_{1}},\tfrac{N_{21}}{n_{2}%
},...,\tfrac{N_{I1}}{n_{I}})^{T}, \label{piBar}%
\end{equation}
and for\ calculating $\widetilde{\boldsymbol{\pi}}$\ the following PAVA
algorithm can be used.

\begin{algorithm}
[Order restricted estimation of probabilities]\label{AlgorP}The MLE of
$\boldsymbol{\pi}=(\pi_{1},\pi_{2},\cdots,\pi_{I})^{T}$ under the restriction
of $\pi_{1}\leq\pi_{2}\leq\cdots\leq\pi_{I}$, $\widetilde{\boldsymbol{\pi}%
}=(\widetilde{\pi}_{1},\widetilde{\pi}_{2},\cdots,\widetilde{\pi}_{I})^{T}$,
is calculated in the following way:

\noindent\texttt{STEP 1: Do }$\widetilde{\boldsymbol{\pi}}:=\overline
{\boldsymbol{\pi}}$\texttt{, where }$\overline{\boldsymbol{\pi}}$ is
(\ref{piBar})\texttt{.}\newline\texttt{STEP 2: While not }$\widetilde{\pi}%
_{i}\leq\widetilde{\pi}_{i+1}$\texttt{ }$\forall i=1,...,I-1$\texttt{
do}\newline\texttt{\hspace*{1.75cm}For }$i=1,...,I-1$\texttt{\newline%
\hspace*{2.25cm}If }$\widetilde{\pi}_{i}\not \leq \widetilde{\pi}_{i+1}%
$\texttt{ do }$\widetilde{\pi}_{i}:=\frac{n_{i}}{n}\widetilde{\pi}_{i}%
+\frac{n_{i+1}}{n}\widetilde{\pi}_{i+1}$\texttt{ and }$\widetilde{\pi}%
_{i+1}:=\widetilde{\pi}_{i}$.
\end{algorithm}

\section{Phi-divergence test statistics\label{Sec2}}

The classical order-restricted likelihood ratio test for testing (\ref{J1a})
is given by%
\[
G^{2}=2\sum_{i=1}^{I}\left(  N_{i1}\log\frac{\widetilde{\pi}_{i}}%
{\widehat{\pi}_{0}}+\left(  n_{i}-N_{i1}\right)  \log\frac{1-\widetilde{\pi
}_{i}}{1-\widehat{\pi}_{0}}\right)
\]
(see for instance Mancuso et al (2001)). The Kullback-Leibler divergence
measure between two $2I$-dimensional probability vectors $\boldsymbol{p=}%
\left(  p_{11},p_{12},...,p_{I1},p_{I2}\right)  ^{T}$ and $\boldsymbol{q=}%
\left(  q_{11},q_{12},...,q_{I1},q_{I2}\right)  ^{T}$, is given by
\[
\mathrm{d}_{Kull}(\boldsymbol{p},\boldsymbol{q})=\sum_{i=1}^{I}\left(
p_{i1}\log\tfrac{p_{i1}}{q_{i1}}+p_{i2}\log\tfrac{p_{i2}}{q_{i2}}\right)  .
\]
It is an easy exercise to verify that
\begin{equation}
G^{2}=2n(\mathrm{d}_{Kull}(\overline{\boldsymbol{p}},\boldsymbol{p}%
(\widehat{\boldsymbol{\theta}}))-\mathrm{d}_{Kull}(\overline{\boldsymbol{p}%
},\boldsymbol{p}(\widetilde{\boldsymbol{\theta}}))), \label{G}%
\end{equation}
where%
\begin{align*}
\overline{\boldsymbol{p}}  &  =\boldsymbol{p}(\overline{\boldsymbol{\theta}%
})=(\tfrac{n_{1}}{n}\overline{\pi}_{1},\tfrac{n_{1}}{n}(1-\overline{\pi}%
_{1}),\tfrac{n_{2}}{n}\overline{\pi}_{2},\tfrac{n_{2}}{n}(1-\overline{\pi}%
_{2}),...,\tfrac{n_{I}}{n}\overline{\pi}_{I},\tfrac{n_{I}}{n}(1-\overline{\pi
}_{I}))^{T}\\
&  =(\tfrac{N_{11}}{n},\tfrac{N_{12}}{n},\tfrac{N_{21}}{n},\tfrac{N_{22}}%
{n},...,\tfrac{N_{I1}}{n},\tfrac{N_{I2}}{n})^{T},\\
\widetilde{\boldsymbol{p}}  &  =\boldsymbol{p}(\widetilde{\boldsymbol{\theta}%
})=(\tfrac{n_{1}}{n}\widetilde{\pi}_{1},\tfrac{n_{1}}{n}(1-\widetilde{\pi}%
_{1}),\tfrac{n_{2}}{n}\widetilde{\pi}_{2},\tfrac{n_{2}}{n}(1-\widetilde{\pi
}_{2}),...,\tfrac{n_{I}}{n}\widetilde{\pi}_{I},\tfrac{n_{I}}{n}%
(1-\widetilde{\pi}_{I}))^{T},\\
\widehat{\boldsymbol{p}}  &  =\boldsymbol{p}(\widehat{\boldsymbol{\theta}%
})=(\tfrac{n_{1}}{n}\widehat{\pi}_{0},\tfrac{n_{1}}{n}(1-\widehat{\pi}%
_{0}),\tfrac{n_{2}}{n}\widehat{\pi}_{0},\tfrac{n_{2}}{n}(1-\widehat{\pi}%
_{0}),...,\tfrac{n_{I}}{n}\widehat{\pi}_{0},\tfrac{n_{I}}{n}(1-\widehat{\pi
}_{0}))^{T}.
\end{align*}

The classical order-restricted chi-square test statistic for testing
(\ref{J1a}), known as Bartholomew's test-statistic, is given by
\begin{equation}
X^{2}=\frac{1}{\widehat{\pi}_{0}\left(  1-\widehat{\pi}_{0}\right)  }%
{\displaystyle\sum\limits_{i=1}^{I}}
n_{i}\left(  \widetilde{\pi}_{i}-\widehat{\pi}_{0}\right)  ^{2}, \label{J2}%
\end{equation}
which can be written as
\begin{equation}
X^{2}=2n\mathrm{d}_{Pearson}(\boldsymbol{p}(\widetilde{\boldsymbol{\theta}%
}),\boldsymbol{p}(\widehat{\boldsymbol{\theta}})), \label{J3}%
\end{equation}
where $\mathrm{d}_{Pearson}(\boldsymbol{p},\boldsymbol{q})$ is the Pearson
divergence measure defined by%
\[
\mathrm{d}_{Pearson}(\boldsymbol{p},\boldsymbol{q})=\frac{1}{2}\sum_{i=1}%
^{I}\left(  \tfrac{(p_{i1}-q_{i1})^{2}}{q_{i1}}+\tfrac{(p_{i2}-q_{i2})^{2}%
}{q_{i2}}\right)  .
\]
Details about this test-statistic can be found in Fleiss et al. (2003, Section 9.3).

More general than the Kullback-Leibler divergence and Pearson divergence
measures are $\phi$-divergence measures, defined as
\[
d_{\phi}(\boldsymbol{p},\boldsymbol{q})=\sum_{i=1}^{I}\left(  q_{i1}%
\phi\left(  \tfrac{p_{i1}}{q_{i1}}\right)  +q_{i2}\phi\left(  \tfrac{p_{i2}%
}{q_{i2}}\right)  \right)  ,
\]
where $\phi:%
\mathbb{R}
_{+}\longrightarrow%
\mathbb{R}
$ is a convex function such that $\phi(1)=\phi^{\prime}(1)=0$, $\phi
^{\prime\prime}(1)>0$, $0\phi(\frac{0}{0})=0$, $0\phi(\frac{p}{0}%
)=p\lim_{u\rightarrow\infty}\frac{\phi(u)}{u}$, for $p\neq0$. For more details
about $\phi$-divergence measures see Pardo (2006).

Based on $\phi$-divergence measures we shall consider in this paper two
families of order-restricted $\phi$-divergence test statistics valid for
testing (\ref{J1a}) or (\ref{J1b}). The first one generalizes the
order-restricted likelihood ratio test given in (\ref{G}) in the sense that we
replace on it the Kullback-Leibler divergence measure by a phi-divergence
measure and its expression is
\begin{align}
&  T_{\phi}(\overline{\boldsymbol{p}},\boldsymbol{p}%
(\widetilde{\boldsymbol{\theta}}),\boldsymbol{p}(\widehat{\boldsymbol{\theta}%
}))=\frac{2n}{\phi^{\prime\prime}(1)}(\mathrm{d}_{\phi}(\overline
{\boldsymbol{p}},\boldsymbol{p}(\widehat{\boldsymbol{\theta}}))-\mathrm{d}%
_{\phi}(\overline{\boldsymbol{p}},\boldsymbol{p}(\widetilde{\boldsymbol{\theta
}})))\label{J4}\\
&  =\frac{2}{\phi^{\prime\prime}(1)}\left\{  \sum_{i=1}^{I}n_{i}\left(
\widehat{\pi}_{0}\phi\left(  \frac{N_{i1}}{n_{i}\widehat{\pi}_{0}}\right)
-\widetilde{\pi}_{i}\phi\left(  \frac{N_{i1}}{n_{i}\widetilde{\pi}_{i}%
}\right)  +(1-\widehat{\pi}_{0})\phi\left(  \frac{n_{i}-N_{i1}}{n_{i}%
(1-\widehat{\pi}_{0})}\right)  -\left(  1-\widetilde{\pi}_{i}\right)
\phi\left(  \frac{n_{i}-N_{i1}}{n_{i}\left(  1-\widetilde{\pi}_{i}\right)
}\right)  \right)  \right\}  .\nonumber
\end{align}
For $\phi(x)=x\log x-x+1$, we get the likelihood ratio test.

The second one generalizes the order-restricted Pearson test statistic given
in (\ref{J3}) in the sense that we replace on it the Pearson divergence
measure by a phi-divergence measure and its expression is
\begin{align}
S_{\phi}(\boldsymbol{p}(\widetilde{\boldsymbol{\theta}}),\boldsymbol{p}%
(\widehat{\boldsymbol{\theta}}))  &  =\frac{2n}{\phi^{\prime\prime}(1)}%
d_{\phi}(\boldsymbol{p}(\widetilde{\boldsymbol{\theta}}),\boldsymbol{p}%
(\widehat{\boldsymbol{\theta}}))\label{J5}\\
&  =\frac{2}{\phi^{\prime\prime}(1)}\left\{  \sum_{i=1}^{I}n_{i}\left(
\widehat{\pi}_{0}\phi\left(  \frac{\widetilde{\pi}_{i}}{\widehat{\pi}_{0}%
}\right)  +\left(  1-\widehat{\pi}_{0}\right)  \phi\left(  \frac
{1-\widetilde{\pi}_{i}}{1-\widehat{\pi}_{0}}\right)  \right)  \right\}
.\nonumber
\end{align}
For $\phi(x)=\frac{1}{2}\left(  x-1\right)  ^{2}$, we get the Pearson test-statistics.

The following theorem provides the link between the both test-statistics,
$T_{\phi}(\overline{\boldsymbol{p}},\boldsymbol{p}%
(\widetilde{\boldsymbol{\theta}}),\boldsymbol{p}(\widehat{\boldsymbol{\theta}%
}))$ and $S_{\phi}(\boldsymbol{p}(\widetilde{\boldsymbol{\theta}%
}),\boldsymbol{p}(\widehat{\boldsymbol{\theta}}))$, and Wald-type test-statistics.

\begin{theorem}
\label{th1}For testing (\ref{J1a}) or (\ref{J1b}), the asymptotic distribution
of
\[
T\in\{T_{\phi}(\overline{\boldsymbol{p}},\boldsymbol{p}%
(\widetilde{\boldsymbol{\theta}}),\boldsymbol{p}(\widehat{\boldsymbol{\theta}%
})),S_{\phi}(\boldsymbol{p}(\widetilde{\boldsymbol{\theta}}),\boldsymbol{p}%
(\widehat{\boldsymbol{\theta}})),W(\widetilde{\boldsymbol{\theta}%
},\widehat{\boldsymbol{\theta}}),H(\widetilde{\boldsymbol{\theta}%
},\widehat{\boldsymbol{\theta}}),D(\overline{\boldsymbol{\theta}%
},\widetilde{\boldsymbol{\theta}},\widehat{\boldsymbol{\theta}})\}
\]
is common and is given by%
\[
\lim_{n\rightarrow\infty}\Pr(T\leq x)=\sum_{i=0}^{I-1}w_{i}(I-1,\boldsymbol{V}%
)\Pr(\chi_{i}^{2}\leq x),
\]
where%
\begin{equation}
\boldsymbol{V}=\boldsymbol{G}_{I-1}\mathrm{diag}^{-1}(\boldsymbol{\nu}^{\ast
})\boldsymbol{G}_{I-1}^{T}+\frac{1}{\nu_{I}}\boldsymbol{e}_{I-1}%
\boldsymbol{e}_{I-1}^{T}, \label{A}%
\end{equation}
and $\{w_{i}(I-1,\boldsymbol{V})\}_{i=0}^{I-1}$ is the set of weights such
that $\sum_{i=0}^{I-1}w_{i}(I-1,\boldsymbol{V})=1$ and its values are given in
Theorem \ref{prop1}.
\end{theorem}

\begin{proof}
Let $\widehat{\boldsymbol{\theta}}$ be the $I$-dimensional vector given in
(\ref{thetaHat}). The second order Taylor expansion of function $\mathrm{d}%
_{\phi}(\boldsymbol{\theta})=\mathrm{d}_{\phi}(\boldsymbol{p}%
(\boldsymbol{\theta}),\boldsymbol{p}(\widehat{\boldsymbol{\theta}}))$ about
$\widehat{\boldsymbol{\theta}}$ is%
\begin{equation}
\mathrm{d}_{\phi}(\boldsymbol{\theta})=\mathrm{d}_{\phi}%
(\widehat{\boldsymbol{\theta}})+(\boldsymbol{\theta}%
-\widehat{\boldsymbol{\theta}})^{T}\left.  \frac{\partial}{\partial
\boldsymbol{\theta}}\mathrm{d}_{\phi}(\boldsymbol{\theta})\right\vert
_{\boldsymbol{\theta=}\widehat{\boldsymbol{\theta}}}+\frac{1}{2}%
(\boldsymbol{\theta}-\widehat{\boldsymbol{\theta}})^{T}\left.  \frac
{\partial^{2}}{\partial\boldsymbol{\theta}\partial\boldsymbol{\theta}^{T}%
}\mathrm{d}_{\phi}(\boldsymbol{\theta})\right\vert _{\boldsymbol{\theta
=}\widehat{\boldsymbol{\theta}}}(\boldsymbol{\theta}%
-\widehat{\boldsymbol{\theta}})+\mathrm{o}\left(  \left\Vert
\boldsymbol{\theta}-\widehat{\boldsymbol{\theta}}\right\Vert ^{2}\right)  ,
\label{eq16}%
\end{equation}
where%
\begin{align*}
\left.  \frac{\partial}{\partial\boldsymbol{\theta}}\mathrm{d}_{\phi
}(\boldsymbol{\theta})\right\vert _{\boldsymbol{\theta=}%
\widehat{\boldsymbol{\theta}}}  &  =\boldsymbol{0}_{I-1},\\
\left.  \frac{\partial^{2}}{\partial\boldsymbol{\theta}\partial
\boldsymbol{\theta}^{T}}\mathrm{d}_{\phi}(\boldsymbol{\theta})\right\vert
_{\boldsymbol{\theta=}\widehat{\boldsymbol{\theta}}}  &  =\phi^{\prime\prime
}\left(  1\right)  \mathcal{I}_{F}^{(n_{1},...,n_{I})}%
(\widehat{\boldsymbol{\theta}}),
\end{align*}
and
\[
\mathcal{I}_{F}^{(n_{1},...,n_{I})}(\boldsymbol{\theta})=\boldsymbol{X}%
^{T}\mathrm{diag}\{\tfrac{n_{i}}{n}\pi_{i}(1-\pi_{i})\}_{i=1}^{I}%
\boldsymbol{X}.
\]
Let $\overline{\boldsymbol{\theta}}=\boldsymbol{X}^{-1}\mathrm{logit}%
(\overline{\boldsymbol{\pi}})$, where $\overline{\boldsymbol{\pi}}$ is
(\ref{piBar}). In particular, for $\boldsymbol{\theta=}\overline
{\boldsymbol{\theta}}$ we have%
\[
\mathrm{d}_{\phi}(\boldsymbol{p}(\overline{\boldsymbol{\theta}}%
),\boldsymbol{p}(\widehat{\boldsymbol{\theta}}))=\frac{\phi^{\prime\prime
}\left(  1\right)  }{2}(\overline{\boldsymbol{\theta}}%
-\widehat{\boldsymbol{\theta}})^{T}\mathcal{I}_{F}^{(n_{1},...,n_{I}%
)}(\widehat{\boldsymbol{\theta}})(\overline{\boldsymbol{\theta}}%
-\widehat{\boldsymbol{\theta}})+\mathrm{o}\left(  \left\Vert \overline
{\boldsymbol{\theta}}-\widehat{\boldsymbol{\theta}}\right\Vert ^{2}\right)
\]
and for $\boldsymbol{\theta=}\widetilde{\boldsymbol{\theta}}$%
\[
\mathrm{d}_{\phi}(\boldsymbol{p}(\widetilde{\boldsymbol{\theta}}%
),\boldsymbol{p}(\widehat{\boldsymbol{\theta}}))=\frac{\phi^{\prime\prime
}\left(  1\right)  }{2}(\widetilde{\boldsymbol{\theta}}%
-\widehat{\boldsymbol{\theta}})^{T}\mathcal{I}_{F}^{(n_{1},...,n_{I}%
)}(\widehat{\boldsymbol{\theta}})(\widetilde{\boldsymbol{\theta}%
}-\widehat{\boldsymbol{\theta}})+\mathrm{o}\left(  \left\Vert
\widetilde{\boldsymbol{\theta}}-\widehat{\boldsymbol{\theta}}\right\Vert
^{2}\right)  ,
\]
and then taking into account that $\lim_{n\rightarrow\infty}\mathcal{I}%
_{F}^{(n_{1},...,n_{I})}(\boldsymbol{\theta})=\mathcal{I}_{F}%
(\boldsymbol{\theta})$,%
\begin{align*}
T_{\phi}(\overline{\boldsymbol{p}},\boldsymbol{p}%
(\widetilde{\boldsymbol{\theta}}),\boldsymbol{p}(\widehat{\boldsymbol{\theta}%
}))  &  =\frac{2n}{\phi^{\prime\prime}(1)}(\mathrm{d}_{\phi}(\overline
{\boldsymbol{p}},\boldsymbol{p}(\widehat{\boldsymbol{\theta}}))-\mathrm{d}%
_{\phi}(\overline{\boldsymbol{p}},\boldsymbol{p}(\widetilde{\boldsymbol{\theta
}})))\\
&  =n(\overline{\boldsymbol{\theta}}-\widehat{\boldsymbol{\theta}}%
)^{T}\mathcal{I}_{F}(\widehat{\boldsymbol{\theta}})(\overline
{\boldsymbol{\theta}}-\widehat{\boldsymbol{\theta}})-n(\overline
{\boldsymbol{\theta}}-\widetilde{\boldsymbol{\theta}})^{T}\mathcal{I}%
_{F}(\widetilde{\boldsymbol{\theta}})(\overline{\boldsymbol{\theta}%
}-\widetilde{\boldsymbol{\theta}})+\mathrm{o}_{P}\left(  1\right) \\
&  =D(\overline{\boldsymbol{\theta}},\widetilde{\boldsymbol{\theta}%
},\widehat{\boldsymbol{\theta}})+\mathrm{o}_{P}\left(  1\right)  .
\end{align*}
In a similar way it is obtained%
\[
\mathrm{d}_{\phi}(\boldsymbol{p}(\overline{\boldsymbol{\theta}}%
),\boldsymbol{p}(\widetilde{\boldsymbol{\theta}}))=\frac{\phi^{\prime\prime
}\left(  1\right)  }{2}(\overline{\boldsymbol{\theta}}%
-\widetilde{\boldsymbol{\theta}})^{T}\mathcal{I}_{F}^{(n_{1},...,n_{I}%
)}(\widetilde{\boldsymbol{\theta}})(\overline{\boldsymbol{\theta}%
}-\widetilde{\boldsymbol{\theta}})+\mathrm{o}\left(  \left\Vert \overline
{\boldsymbol{\theta}}-\widetilde{\boldsymbol{\theta}}\right\Vert ^{2}\right)
,
\]
and then taking into account that $\lim_{n\rightarrow\infty}\mathcal{I}%
_{F}^{(n_{1},...,n_{I})}(\boldsymbol{\theta})=\mathcal{I}_{F}%
(\boldsymbol{\theta})$,%
\begin{align*}
S_{\phi}(\boldsymbol{p}(\widetilde{\boldsymbol{\theta}}),\boldsymbol{p}%
(\widehat{\boldsymbol{\theta}}))  &  =\frac{2n}{\phi^{\prime\prime}(1)}%
d_{\phi}(\boldsymbol{p}(\widetilde{\boldsymbol{\theta}}),\boldsymbol{p}%
(\widehat{\boldsymbol{\theta}}))\\
&  =n(\widetilde{\boldsymbol{\theta}}-\widehat{\boldsymbol{\theta}}%
)^{T}\mathcal{I}_{F}(\widehat{\boldsymbol{\theta}}%
)(\widetilde{\boldsymbol{\theta}}-\widehat{\boldsymbol{\theta}})+\mathrm{o}%
_{P}\left(  1\right) \\
&  =H(\widetilde{\boldsymbol{\theta}},\widehat{\boldsymbol{\theta}%
})+\mathrm{o}_{P}\left(  1\right)
\end{align*}
According to Proposition 4.4.1 in Silvapulle and Sen (2005)%
\[
G^{2}=D(\overline{\boldsymbol{\theta}},\widetilde{\boldsymbol{\theta}%
},\widehat{\boldsymbol{\theta}})+\mathrm{o}_{P}\left(  1\right)
=H(\widetilde{\boldsymbol{\theta}},\widehat{\boldsymbol{\theta}}%
)+\mathrm{o}_{P}\left(  1\right)  =W(\widetilde{\boldsymbol{\theta}%
},\widehat{\boldsymbol{\theta}})+\mathrm{o}_{P}\left(  1\right)  ,
\]
which means that the asymptotic distribution of
\[
T\in\{T_{\phi}(\overline{\boldsymbol{p}},\boldsymbol{p}%
(\widetilde{\boldsymbol{\theta}}),\boldsymbol{p}(\widehat{\boldsymbol{\theta}%
})),S_{\phi}(\boldsymbol{p}(\widetilde{\boldsymbol{\theta}}),\boldsymbol{p}%
(\widehat{\boldsymbol{\theta}})),W(\widetilde{\boldsymbol{\theta}%
},\widehat{\boldsymbol{\theta}}),H(\widetilde{\boldsymbol{\theta}%
},\widehat{\boldsymbol{\theta}}),D(\overline{\boldsymbol{\theta}%
},\widetilde{\boldsymbol{\theta}},\widehat{\boldsymbol{\theta}})\}
\]
is common. Such a distribution can be established from the likelihood ratio
test-statistic used for the problem formulated in (6.13) of Silvapulle and Sen
(2005)%
\[
\lim_{n\rightarrow\infty}\Pr(G^{2}\leq x)=\sum_{i=0}^{I-1}w_{i}%
(I-1,\boldsymbol{V})\Pr(\chi_{i}^{2}\leq x),
\]
where%
\[
\boldsymbol{V}=\boldsymbol{B}\mathrm{diag}^{-1}(\boldsymbol{\nu}%
)\boldsymbol{B}^{T},
\]
with%
\begin{align*}
\boldsymbol{\nu}  &  =(\boldsymbol{\nu}^{\ast},\nu_{I})^{T}=(\nu_{1}%
,...,\nu_{I})^{T},\\
\boldsymbol{B}  &  =(\boldsymbol{G}_{I-1},-\boldsymbol{e}_{I-1}).
\end{align*}
Using the partitioned structure of $\boldsymbol{B}$\ and $\boldsymbol{\nu}$,
(\ref{A}) is obtained.
\end{proof}

The following result is based on the third way for computation of weights
given in page 79 of Silvapulle and Sen (2005).

\begin{theorem}
\label{prop1}The set of weights $\{w_{i}(I-1,\boldsymbol{V})\}_{i=0}^{I-1}$ of
the asymptotic distribution given in Theorem \ref{th1} is computed as follows%
\begin{equation}
w_{i}(I-1,\boldsymbol{V})=\Pr\left(  \arg\min_{\zeta\in%
\mathbb{R}
_{+}^{I-1}}(\boldsymbol{Z}-\zeta)^{T}\boldsymbol{V}^{-1}(\boldsymbol{Z}%
-\zeta)\in%
\mathbb{R}
_{+}^{I-1}(i)\right)  ,\quad i=0,...,I-1 \label{w}%
\end{equation}
$\boldsymbol{Z}\sim\mathcal{N}_{I-1}(\boldsymbol{0}_{I-1},\boldsymbol{V})$, $%
\mathbb{R}
_{+}^{I-1}=\{\zeta\in%
\mathbb{R}
_{+}^{I-1}:\zeta\geq0\}$ and $%
\mathbb{R}
_{+}^{I-1}(i)\subset%
\mathbb{R}
_{+}^{I-1}$ such that $i$ components are strictly positive and $I-1-i$
components are null. In particular, for $I\in\{2,3,4\}$ (\ref{w}) has the
explicit expressions given in (3.24), (3.25) and (3.26) of Silvapulle and Sen (2005).
\end{theorem}

For computing (\ref{w}) is useful to know the following explicit expression
$\boldsymbol{V}^{-1}=\boldsymbol{T}_{I-1}^{T}\boldsymbol{\Sigma}%
_{\boldsymbol{\nu}^{\ast}}\boldsymbol{T}_{I-1}$, where $\boldsymbol{T}%
_{I-1}=\boldsymbol{G}_{I-1}^{-1}$ is an upper triangular matrix of $1$-s,
$\boldsymbol{\Sigma}_{\boldsymbol{\nu}^{\ast}}=\mathrm{diag}(\boldsymbol{\nu
}^{\ast})-\boldsymbol{\nu}^{\ast}(\boldsymbol{\nu}^{\ast})^{T}$, and to
simulate the probability according to the following algorithm. For simulation
\begin{align}
\widehat{\boldsymbol{V}}  &  =\boldsymbol{G}_{I-1}\mathrm{diag}^{-1}%
(\widehat{\boldsymbol{\nu}}^{\ast})\boldsymbol{G}_{I-1}^{T}+\frac
{1}{\widehat{\nu}_{I}}\boldsymbol{e}_{I-1}\boldsymbol{e}_{I-1}^{T}%
,\label{VHat}\\
\widehat{\boldsymbol{V}}^{-1}  &  =\boldsymbol{T}_{I-1}^{T}\boldsymbol{\Sigma
}_{\widehat{\boldsymbol{\nu}}^{\ast}}\boldsymbol{T}_{I-1}, \label{VHarInv}%
\end{align}
\ are needed rather than $\boldsymbol{V}$ and $\boldsymbol{V}^{-1}$\ respectively.

\begin{algorithm}
[Estimation of weights]\label{AlgorW}The estimators of the weights given in
Theorem \ref{prop1}, $\widehat{w}_{i}=w_{i}(I-1,\widehat{\boldsymbol{V}})$,
are obtained by Monte Carlo in the following way:

\noindent\texttt{STEP 1: For }$i=0,...,I-1$\texttt{, do }$N(i):=0$%
\texttt{.}\newline\texttt{STEP 2: Repeat the following steps }$R$\texttt{ (say
}$R=1,000,000$\texttt{) times:}\newline\hspace*{0.75cm}\texttt{STEP 2.1:
Generate an observation, }$\boldsymbol{z}$\texttt{, from }$\boldsymbol{Z\sim
}\mathcal{N}_{I-1}(\boldsymbol{0}_{I-1},\widehat{\boldsymbol{V}})$\texttt{.
E.g., the\newline\hspace*{1cm}\hspace*{1cm}\hspace*{0.75cm}NAG Fortran library
subroutines G05CBF, G05EAF, and G05EZF can be useful.\newline\hspace
*{0.75cm}STEP 2.2: Compute }$\widehat{\boldsymbol{\zeta}}(\boldsymbol{z}%
)\boldsymbol{=}\arg\min_{\boldsymbol{\zeta\in}%
\mathbb{R}
_{+}^{I-1}}\tfrac{1}{2}\boldsymbol{\zeta}^{T}\widehat{\boldsymbol{V}}%
^{-1}\boldsymbol{\zeta}-(\widehat{\boldsymbol{V}}^{-1}\boldsymbol{z}%
)^{T}\boldsymbol{\zeta}$.\texttt{ E.g., the IMSL Fortran \newline\hspace
*{1cm}\hspace*{1cm}\hspace*{0.75cm}library subroutine DQPROG can be
useful.\newline\hspace*{0.75cm}STEP 2.3: Count }$i^{\ast}$\texttt{, the number
of strictly positive components contained in }$\widehat{\boldsymbol{\zeta}%
}(\boldsymbol{z})$\texttt{, and \newline\hspace*{1cm}\hspace*{1cm}%
\hspace*{0.75cm}do }$N(i^{\ast}):=N(i^{\ast})+1$.\newline\texttt{STEP 3: Do
}$\widehat{w}_{i}:=\frac{N(i)}{R}$ for $i=0,...,I-1$.
\end{algorithm}

\section{Example\label{Sec3}}

In this section the data set of the introduction (Table \ref{tt1}) is
analyzed. The sample, a realization of $\boldsymbol{N}$, is summarized in the
following vector%
\begin{align*}
\boldsymbol{n}  &  =(n_{11},n_{12},n_{21},n_{22},n_{31},n_{32},n_{41}%
,n_{42})^{T}\\
&  =(48,17066,38,14464,5,788,2,163)^{T}.
\end{align*}
The estimated vectors of interest are%
\begin{align*}
\overline{\boldsymbol{\pi}}  &  =\left(  \tfrac{48}{17114},\tfrac{38}%
{14502},\tfrac{5}{793},\tfrac{2}{165}\right)  ^{T}%
=(0.0028,0.0026,0.0063,0.0121)^{T},\\
\widetilde{\boldsymbol{\pi}}  &  =\left(  \tfrac{43}{15\,808},\tfrac
{43}{15\,808},\tfrac{5}{793},\tfrac{2}{165}\right)  ^{T}%
=(0.0027,0.0027,0.0063,0.0121)^{T},\\
\widehat{\pi}_{0}  &  =\tfrac{93}{32\,574}=0.0029,
\end{align*}
and%
\begin{align*}
\overline{\boldsymbol{\theta}}  &  =\boldsymbol{X}^{-1}\mathrm{logit}%
(\overline{\boldsymbol{\pi}})=(-4.\,\allowbreak400\,6,-1.\,\allowbreak
473,-1.\,\allowbreak541\,2,-0.659\,46)^{T},\\
\widetilde{\boldsymbol{\theta}}  &  =\boldsymbol{X}^{-1}\mathrm{logit}%
(\widetilde{\boldsymbol{\pi}})=(-4.\,\allowbreak400\,6,-1.\,\allowbreak
503\,7,-1.\,\allowbreak503\,7,-0.659\,46)^{T}\\
\widehat{\boldsymbol{\theta}}  &  =(\mathrm{logit}(\widehat{\pi}%
_{0}),0,0,0)^{T}=(-5.\allowbreak8558,0,0,0)^{T}.
\end{align*}
For the asymptotic distribution the weighs can be calculated though%
\[
\widehat{\boldsymbol{V}}=%
\begin{pmatrix}
83774.0156250 & -14428.7177734 & 0\\
-14428.7177734 & 15217.7109375 & -788.9927979\\
0 & -788.9927979 & 1457.5666504
\end{pmatrix}
,
\]
calculating the correlation coefficients%
\[
\widehat{\rho}_{12}=-0.40411,\quad\widehat{\rho}_{13}=0,\quad\widehat{\rho
}_{23}=-0.16753,
\]
the partial correlation coefficients%
\[
\widehat{\rho}_{12\bullet3}=-0.4099,\quad\widehat{\rho}_{13\bullet
2}=-0.07507\,2,\quad\widehat{\rho}_{23\bullet1}=-0.183\,15,
\]
and evaluating the following expressions%
\begin{align*}
\widehat{w}_{3}  &  =\frac{1}{4\pi}\left(  2\pi-\arccos\left(  \widehat{\rho
}_{12}\right)  -\arccos\left(  \widehat{\rho}_{13}\right)  -\arccos\left(
\widehat{\rho}_{23}\right)  \right)  =0.07850,\\
\widehat{w}_{2}  &  =\frac{1}{4\pi}\left(  3\pi-\arccos(\widehat{\rho
}_{12\bullet3})-\arccos(\widehat{\rho}_{13\bullet2})-\arccos(\widehat{\rho
}_{23\bullet1})\right)  =0.32075,\\
\widehat{w}_{1}  &  =0.5-w_{0}(\widehat{\boldsymbol{\theta}}%
)=0.5-0.07850=0.4215,\\
\widehat{w}_{0}  &  =0.5-w_{1}(\widehat{\boldsymbol{\theta}}%
)=0.5-0.32075=0.17925.
\end{align*}

If we take $\phi_{\lambda}(x)=\frac{1}{\lambda(1+\lambda)}(x^{\lambda
+1}-x-\lambda(x-1))$, where for each $\lambda\in%
\mathbb{R}
-\{-1,0\}$, the \textquotedblleft power divergence family\textquotedblright%
\ is obtained%
\begin{equation}
d_{\lambda}(\boldsymbol{p},\boldsymbol{q})=\frac{1}{\lambda(\lambda+1)}\left(
%
{\displaystyle\sum\limits_{i=1}^{I}}
{\displaystyle\sum\limits_{j=1}^{J}}
\frac{p_{ij}^{\lambda+1}}{q_{ij}^{\lambda}}-1\right)  \text{, for each
}\lambda\in%
\mathbb{R}
-\{-1,0\}\text{.} \label{CR}%
\end{equation}
It is also possible to cover the real line for $\lambda$, by defining
$d_{\lambda}(\boldsymbol{p},\boldsymbol{q})=\lim_{t\rightarrow\lambda}%
d_{t}(\boldsymbol{p},\boldsymbol{q})$, for $\lambda\in\{-1,0\}$. It is well
known that $d_{0}(\boldsymbol{p},\boldsymbol{q})=d_{Kull}(\boldsymbol{p}%
,\boldsymbol{q})$ and $d_{1}(\boldsymbol{p},\boldsymbol{q})=d_{Pearson}%
(\boldsymbol{p},\boldsymbol{q})$. This is very interesting since this means
that the power divergence based family of test-statistics contain as special
cases $G^{2}$ and $X^{2}$.

Finally, the expressions of the test-statistics are summarized in Table
\ref{tEst}. It can be seen that the null hypothesis cannot be rejected for
$W(\widetilde{\boldsymbol{\theta}},\widehat{\boldsymbol{\theta}})$,
$H(\widetilde{\boldsymbol{\theta}},\widehat{\boldsymbol{\theta}})$,
$D(\overline{\boldsymbol{\theta}},\widetilde{\boldsymbol{\theta}%
},\widehat{\boldsymbol{\theta}})$, $T_{-1.5}(\overline{\boldsymbol{p}%
},\boldsymbol{p}(\widetilde{\boldsymbol{\theta}}),\boldsymbol{p}%
(\widehat{\boldsymbol{\theta}}))$, $T_{-1}(\overline{\boldsymbol{p}%
},\boldsymbol{p}(\widetilde{\boldsymbol{\theta}}),\boldsymbol{p}%
(\widehat{\boldsymbol{\theta}}))$, $T_{-0.5}(\overline{\boldsymbol{p}%
},\boldsymbol{p}(\widetilde{\boldsymbol{\theta}}),\boldsymbol{p}%
(\widehat{\boldsymbol{\theta}}))$, $S_{-1.5}(\boldsymbol{p}%
(\widetilde{\boldsymbol{\theta}}),\boldsymbol{p}(\widehat{\boldsymbol{\theta}%
}))$, $S_{-1}(\boldsymbol{p}(\widetilde{\boldsymbol{\theta}}),\boldsymbol{p}%
(\widehat{\boldsymbol{\theta}}))$, $S_{-0.5}(\boldsymbol{p}%
(\widetilde{\boldsymbol{\theta}}),\boldsymbol{p}(\widehat{\boldsymbol{\theta}%
}))$ and should be rejected for $T_{0}(\overline{\boldsymbol{p}}%
,\boldsymbol{p}(\widetilde{\boldsymbol{\theta}}),\boldsymbol{p}%
(\widehat{\boldsymbol{\theta}}))$, $T_{\frac{2}{3}}(\overline{\boldsymbol{p}%
},\boldsymbol{p}(\widetilde{\boldsymbol{\theta}}),\boldsymbol{p}%
(\widehat{\boldsymbol{\theta}}))$, $T_{1}(\overline{\boldsymbol{p}%
},\boldsymbol{p}(\widetilde{\boldsymbol{\theta}}),\boldsymbol{p}%
(\widehat{\boldsymbol{\theta}}))$, $S_{0}(\boldsymbol{p}%
(\widetilde{\boldsymbol{\theta}}),\boldsymbol{p}(\widehat{\boldsymbol{\theta}%
}))$, $S_{\frac{2}{3}}(\boldsymbol{p}(\widetilde{\boldsymbol{\theta}%
}),\boldsymbol{p}(\widehat{\boldsymbol{\theta}}))$, $S_{1}(\boldsymbol{p}%
(\widetilde{\boldsymbol{\theta}}),\boldsymbol{p}(\widehat{\boldsymbol{\theta}%
}))$. Even though the sample size seems to be large enough, this is a case of
small values of $\pi_{1}$, $\pi_{2}$, $\pi_{3}$, $\pi_{4}$ which is known to
have not reliable behavior in the values calculated for the $p$-values in
order to make decisions. In the simulation study we shall study such a case
and according to the results the rejection of the null hypothesis is supported
since with $\{T_{\lambda}(\overline{\boldsymbol{p}},\boldsymbol{p}%
(\widetilde{\boldsymbol{\theta}}),\boldsymbol{p}(\widehat{\boldsymbol{\theta}%
}))\}_{\lambda\in\{0,\frac{2}{3},1\}}$\ are obtained the most realiable
test-statistics. As conclussion, an increase in maternal alcohol consumption
is associated with an increase in the probability of malformation.%

\begin{table}[htbp]  \tiny\tabcolsep2.8pt  \centering
$\
\begin{tabular}
[c]{c}%
\begin{tabular}
[c]{ccccccc}\hline\hline
test-statistic & $\hspace*{0.5cm}\lambda=-1.5\hspace*{0.5cm}$ & $\hspace
*{0.5cm}\lambda=-1\hspace*{0.5cm}$ & $\hspace*{0.5cm}\lambda=-0.5\hspace
*{0.5cm}$ & $\hspace*{0.5cm}\lambda=0\hspace*{0.5cm}$ & $\hspace
*{0.5cm}\lambda=\frac{2}{3}\hspace*{0.5cm}$ & $\hspace*{0.5cm}\lambda
=1\hspace*{0.5cm}$\\\hline
\multicolumn{1}{l}{$\overset{}{T}_{\lambda}(\overline{\boldsymbol{p}%
},\boldsymbol{p}(\widetilde{\boldsymbol{\theta}}),\boldsymbol{p}%
(\widehat{\boldsymbol{\theta}}))$} & $3.3068$ & $3.8173$ & $4.4920$ & $5.4057$
& $7.2076$ & $8.4895$\\
$p$\textrm{-}$\mathrm{value}(T_{\lambda})$ & $0.1177$ & $0.0911$ & $0.0650$ &
$0.0413$ & $0.0169$ & $0.0090$\\
\multicolumn{1}{l}{$\overset{}{S}_{\lambda}(\boldsymbol{p}%
(\widetilde{\boldsymbol{\theta}}),\boldsymbol{p}(\widehat{\boldsymbol{\theta}%
}))$} & $3.2993$ & $3.8124$ & $4.4896$ & $5.4057$ & $7.2107$ & $8.4942$\\
$p$\textrm{-}$\mathrm{value}(S_{\lambda})$ & $0.1181$ & $0.0913$ & $0.0651$ &
$0.0413$ & $0.0169$ & $0.0090$\\\hline\hline
\end{tabular}
\bigskip\\%
\begin{tabular}
[c]{cccccc}\hline\hline
$\overset{}{W}(\widetilde{\boldsymbol{\theta}},\widehat{\boldsymbol{\theta}})$
& $2.5979$ & $H(\widetilde{\boldsymbol{\theta}},\widehat{\boldsymbol{\theta}%
})$ & $2.6363$ & $D(\overline{\boldsymbol{\theta}}%
,\widetilde{\boldsymbol{\theta}},\widehat{\boldsymbol{\theta}})$ & $2.6462$\\
$p$\textrm{-}$\mathrm{value}(W(\widetilde{\boldsymbol{\theta}}%
,\widehat{\boldsymbol{\theta}}))$ & $0.1686$ & $p$\textrm{-}$\mathrm{value}%
(H(\widetilde{\boldsymbol{\theta}},\widehat{\boldsymbol{\theta}}))$ & $0.1653$
& $p$\textrm{-}$\mathrm{value}(D(\overline{\boldsymbol{\theta}}%
,\widetilde{\boldsymbol{\theta}},\widehat{\boldsymbol{\theta}}))$ &
$0.1645$\\\hline\hline
\end{tabular}
\end{tabular}
\ \ \ \ $%
\caption{Power divergence based test-statistics, Wald type statistics and their corresponding asymptotic p-values.\label{tEst}}%
\end{table}%

\section{Simulation study\label{Sec4}}

For testing (\ref{J1}), by considering $I=4$ binomial random variables, the
following scenarios will be considered:

\begin{itemize}
\item scenario \textbf{A }(small/big proportions): $n_{1}=40$, $n_{2}=30$,
$n_{3}=20$, $n_{4}=10$.

\begin{itemize}
\item scenario \textbf{A-0}: $\pi_{1}=\pi_{2}=\pi_{3}=\pi_{4}=0.05$.

\item scenario \textbf{A-1}: $\pi_{1}=0.05$, $\pi_{2}=\pi_{3}=\pi_{4}=0.1$.

\item scenario \textbf{A-2}: $\pi_{1}=0.05$, $\pi_{2}=0.1$, $\pi_{3}=\pi
_{4}=0.125$.

\item scenario \textbf{A-3}: $\pi_{1}=0.05$, $\pi_{2}=0.1$, $\pi_{3}=0.125$,
$\pi_{4}=0.135$.
\end{itemize}

\item scenario \textbf{B }(small/big proportions): $n_{1}=60$, $n_{2}=45$,
$n_{3}=30$, $n_{4}=15$.

\begin{itemize}
\item scenario \textbf{B-0}: $\pi_{1}=\pi_{2}=\pi_{3}=\pi_{4}=0.05$.

\item scenario \textbf{B-1}: $\pi_{1}=0.05$, $\pi_{2}=\pi_{3}=\pi_{4}=0.1$.

\item scenario \textbf{B-2}: $\pi_{1}=0.05$, $\pi_{2}=0.1$, $\pi_{3}=\pi
_{4}=0.125$.

\item scenario \textbf{B-3}: $\pi_{1}=0.05$, $\pi_{2}=0.1$, $\pi_{3}=0.125$,
$\pi_{4}=0.135$.
\end{itemize}

\item scenario \textbf{C }(small/big proportions): $n_{1}=100$, $n_{2}=75$,
$n_{3}=50$, $n_{4}=25$.

\begin{itemize}
\item scenario \textbf{C-0}: $\pi_{1}=\pi_{2}=\pi_{3}=\pi_{4}=0.05$.

\item scenario \textbf{C-1}: $\pi_{1}=0.05$, $\pi_{2}=\pi_{3}=\pi_{4}=0.1$.

\item scenario \textbf{C-2}: $\pi_{1}=0.05$, $\pi_{2}=0.1$, $\pi_{3}=\pi
_{4}=0.125$.

\item scenario \textbf{C-3}: $\pi_{1}=0.05$, $\pi_{2}=0.1$, $\pi_{3}=0.125$,
$\pi_{4}=0.135$.
\end{itemize}

\item scenario \textbf{D }(intermediate proportions): $n_{1}=40$, $n_{2}=30$,
$n_{3}=20$, $n_{4}=10$.

\begin{itemize}
\item scenario \textbf{D-0}: $\pi_{1}=\pi_{2}=\pi_{3}=\pi_{4}=0.35$.

\item scenario \textbf{D-1}: $\pi_{1}=0.35$, $\pi_{2}=\pi_{3}=\pi_{4}=0.45$.

\item scenario \textbf{D-2}: $\pi_{1}=0.35$, $\pi_{2}=0.45$, $\pi_{3}=\pi
_{4}=0.475$.

\item scenario \textbf{D-3}: $\pi_{1}=0.35$, $\pi_{2}=0.45$, $\pi_{3}=0.475$,
$\pi_{4}=0.485$.
\end{itemize}

\item scenario \textbf{E }(intermediate proportions): $n_{1}=60$, $n_{2}=45$,
$n_{3}=30$, $n_{4}=15$.

\begin{itemize}
\item scenario \textbf{E-0}: $\pi_{1}=\pi_{2}=\pi_{3}=\pi_{4}=0.35$.

\item scenario \textbf{E-1}: $\pi_{1}=0.35$, $\pi_{2}=\pi_{3}=\pi_{4}=0.45$.

\item scenario \textbf{E-2}: $\pi_{1}=0.35$, $\pi_{2}=0.45$, $\pi_{3}=\pi
_{4}=0.475$.

\item scenario \textbf{E-3}: $\pi_{1}=0.35$, $\pi_{2}=0.45$, $\pi_{3}=0.475$,
$\pi_{4}=0.485$.
\end{itemize}

\item scenario \textbf{F }(intermediate proportions): $n_{1}=100$, $n_{2}=75$,
$n_{3}=50$, $n_{4}=25$.

\begin{itemize}
\item scenario \textbf{F-0}: $\pi_{1}=\pi_{2}=\pi_{3}=\pi_{4}=0.35$.

\item scenario \textbf{F-1}: $\pi_{1}=0.35$, $\pi_{2}=\pi_{3}=\pi_{4}=0.45$.

\item scenario \textbf{F-2}: $\pi_{1}=0.35$, $\pi_{2}=0.45$, $\pi_{3}=\pi
_{4}=0.475$.

\item scenario \textbf{F-3}: $\pi_{1}=0.35$, $\pi_{2}=0.45$, $\pi_{3}=0.475$,
$\pi_{4}=0.485$.
\end{itemize}
\end{itemize}

The simulation experiment is performed with $R=50000$ replications\ and in
each of them, apart from the Wald type test-statistics $T^{(h)}\in
\{W^{(h)}(\widetilde{\boldsymbol{\theta}},\widehat{\boldsymbol{\theta}})$,
$H^{(h)}(\widetilde{\boldsymbol{\theta}},\widehat{\boldsymbol{\theta}})$,
$D^{(h)}(\overline{\boldsymbol{\theta}},\widetilde{\boldsymbol{\theta}%
},\widehat{\boldsymbol{\theta}})\}$, $h=1,...,R$, all the power divergence
test statistics, $T^{(h)}\in\{T_{\lambda}^{(h)}(\overline{\boldsymbol{p}%
},\boldsymbol{p}(\widetilde{\boldsymbol{\theta}}),\boldsymbol{p}%
(\widehat{\boldsymbol{\theta}})),S_{\lambda}^{(h)}(\boldsymbol{p}%
(\widetilde{\boldsymbol{\theta}}),\boldsymbol{p}(\widehat{\boldsymbol{\theta}%
}))\}_{\lambda\in I}$, $h=1,...,R$, associated with the interval $I=(-1.5,3)$
are considered. From the p-values it is possible to calculate the proportion
of replications rejected according with the nominal size $\alpha=0.05$, i.e.%
\begin{equation}
\frac{\sum_{h=1}^{R}I\{p\text{-}value(T^{(h)})\leq\alpha\}}{R}, \label{prop}%
\end{equation}
where $I\{\bullet\}$ represents the indicator function. The scenarios ending
in 0 represent that the null hypothesis is true, and are useful to obtain the
simulated significance levels, $\widehat{\alpha}_{T}$, $T\in
\{W(\widetilde{\boldsymbol{\theta}},\widehat{\boldsymbol{\theta}%
}),H(\widetilde{\boldsymbol{\theta}},\widehat{\boldsymbol{\theta}%
}),D(\overline{\boldsymbol{\theta}},\widetilde{\boldsymbol{\theta}%
},\widehat{\boldsymbol{\theta}}),T_{\lambda}(\overline{\boldsymbol{p}%
},\boldsymbol{p}(\widetilde{\boldsymbol{\theta}}),\boldsymbol{p}%
(\widehat{\boldsymbol{\theta}})),S_{\lambda}(\boldsymbol{p}%
(\widetilde{\boldsymbol{\theta}}),\boldsymbol{p}(\widehat{\boldsymbol{\theta}%
}))\}_{\lambda\in I}$, with different sample sizes and kinds of
test-statistics. The scenarios ending in either 1, 2 or 3 represent that the
null hypothesis is false and are useful to obtain the simulated powers,
$\widehat{\beta}_{T}$, with different alternatives, sample sizes and types of
test-statistics. For calculating both, $\widehat{\alpha}_{T}$ and
$\widehat{\beta}_{T}$, (\ref{prop}) is applied, each one in the corresponding scenario.

In Figures \ref{fig1} and \ref{fig2}, $\widehat{\alpha}_{T}$ and
$\widehat{\beta}_{T}$ for all the aforementioned test-statistics are plotted
in different scenarios. The curves represent either $T_{\lambda}%
(\overline{\boldsymbol{p}},\boldsymbol{p}(\widetilde{\boldsymbol{\theta}%
}),\boldsymbol{p}(\widehat{\boldsymbol{\theta}}))$\ or $S_{\lambda
}(\boldsymbol{p}(\widetilde{\boldsymbol{\theta}}),\boldsymbol{p}%
(\widehat{\boldsymbol{\theta}}))$\ power divergence test-statistics, located
respectively on left or right of the panel of plots. The asterisk, square and
circle symbols, represent $W(\widetilde{\boldsymbol{\theta}}%
,\widehat{\boldsymbol{\theta}})$, $H(\widetilde{\boldsymbol{\theta}%
},\widehat{\boldsymbol{\theta}})$ and $D(\overline{\boldsymbol{\theta}%
},\widetilde{\boldsymbol{\theta}},\widehat{\boldsymbol{\theta}})$ respectively
and all of them are repeated on the left as well as on the right in order to
make easier their comparison with $T_{\lambda}(\overline{\boldsymbol{p}%
},\boldsymbol{p}(\widetilde{\boldsymbol{\theta}}),\boldsymbol{p}%
(\widehat{\boldsymbol{\theta}}))$\ or $S_{\lambda}(\boldsymbol{p}%
(\widetilde{\boldsymbol{\theta}}),\boldsymbol{p}(\widehat{\boldsymbol{\theta}%
}))$\ respectively. The black color lines and symbols, representing
$\widehat{\alpha}_{T}$, are useful to select the test-statistics closed to
nominal level According to the criterion given by Dale (1986), a reasonable
exact significance level should verify%
\[
\left\vert \mathrm{logit}(1-\widehat{\alpha}_{T})-\mathrm{logit}%
(1-\alpha)\right\vert \leq\epsilon
\]
for $\epsilon\in\{0.35,0.7\}$, being \textquotedblleft closed to nominal
level\textquotedblright\ the exact significance levels verifying the
inequality with $\epsilon=0.35$ and \textquotedblleft fairly closed to nominal
level\textquotedblright\ the ones with $\epsilon=0.7$. In this study only the
test statistics satisfying the condition with $\epsilon=0.35$ are considered,
and the corresponding upper and lower bounds appear plotted with two
horizontal lines, having in the middle the line associated with the nominal
level, $\alpha=0.05$. Among the test-statistics with simulated significance
levels closed to the nominal level, the test-statistics with higher powers
should be selected but since in general high powers correspond to high
significance levels, this choice is not straightforward. For this reason,
based on $\widehat{\beta}_{T_{0}}-\widehat{\alpha}_{T_{0}}$ or $\widehat{\beta
}_{S_{1}}-\widehat{\alpha}_{S_{1}}$ as baseline, the efficiencies relative to
the likelihood ratio test ($T_{0}(\overline{\boldsymbol{p}},\boldsymbol{p}%
(\widetilde{\boldsymbol{\theta}}),\boldsymbol{p}(\widehat{\boldsymbol{\theta}%
}))=G^{2}$)
\begin{equation}
\rho_{T}=\frac{\left(  \widehat{\beta}_{T}-\widehat{\alpha}_{T}\right)
-\left(  \widehat{\beta}_{T_{0}}-\widehat{\alpha}_{T_{0}}\right)
}{\widehat{\beta}_{T_{0}}-\widehat{\alpha}_{T_{0}}} \label{ro}%
\end{equation}
or the Bartholomew's test ($S_{1}(\boldsymbol{p}(\widetilde{\boldsymbol{\theta
}}),\boldsymbol{p}(\widehat{\boldsymbol{\theta}}))=X^{2}$)
\begin{equation}
\rho_{T}^{\ast}=\frac{\left(  \widehat{\beta}_{T}-\widehat{\alpha}_{T}\right)
-\left(  \widehat{\beta}_{S_{1}}-\widehat{\alpha}_{S_{1}}\right)
}{\widehat{\beta}_{S_{1}}-\widehat{\alpha}_{S_{1}}} \label{ro*}%
\end{equation}
are considered. In Table \ref{table1}, the efficiency of $S_{1}=X^{2}$ is
compared with respect to $T_{0}=G^{2}$, and then if $\rho_{S_{1}}<0$, since
$T_{0}=G^{2}$ is better than $S_{1}=X^{2}$, the plot of the efficiencies in
Figures \ref{fig3} and \ref{fig4} will be only focussed on (\ref{ro}).
Similarly, if $\rho_{S_{1}}>0$, since $T_{0}=G^{2}$ is worse than $S_{1}%
=X^{2}$, the plot of the efficiencies in Figures \ref{fig3} and \ref{fig4}
will be only focussed on (\ref{ro*}).

\begin{center}%
\begin{table}[htbp]  \centering
\begin{tabular}
[c]{ccccccc}\hline
& sc A & sc B & sc C & sc D & sc E & sc F\\\hline
1 & $\rho_{S_{1}}<0$ & $\rho_{S_{1}}<0$ & $\rho_{S_{1}}<0$ & $\rho_{S_{1}}<0$
& $\rho_{S_{1}}<0$ & $\rho_{S_{1}}<0$\\
2 & $\rho_{S_{1}}<0$ & $\rho_{S_{1}}>0$ & $\rho_{S_{1}}<0$ & $\rho_{S_{1}}<0$
& $\rho_{S_{1}}<0$ & $\rho_{S_{1}}<0$\\
3 & $\rho_{S_{1}}<0$ & $\rho_{S_{1}}>0$ & $\rho_{S_{1}}<0$ & $\rho_{S_{1}}<0$
& $\rho_{S_{1}}<0$ & $\rho_{S_{1}}<0$\\\hline
\end{tabular}
\caption{Efficiency of the Bartholomew's test with respect to the likelihood ratio test.\label{table1}}%
\end{table}%

\end{center}

%

\begin{figure}[htbp]  \centering
\begin{tabular}
[c]{cc}%
$\{T_{\lambda}(\overline{\boldsymbol{p}},\boldsymbol{p}%
(\widetilde{\boldsymbol{\theta}}),\boldsymbol{p}(\widehat{\boldsymbol{\theta}%
}))\}_{\lambda\in(-1.5,3)}$ & $\{S_{\lambda}(\boldsymbol{p}%
(\widetilde{\boldsymbol{\theta}}),\boldsymbol{p}(\widehat{\boldsymbol{\theta}%
}))\}_{\lambda\in(-1.5,3)}$\\
\multicolumn{2}{c}{$W(\widetilde{\boldsymbol{\theta}}%
,\widehat{\boldsymbol{\theta}})$, $H(\widetilde{\boldsymbol{\theta}%
},\widehat{\boldsymbol{\theta}})$, $D(\overline{\boldsymbol{\theta}%
},\widetilde{\boldsymbol{\theta}},\widehat{\boldsymbol{\theta}})$}\\
\multicolumn{2}{c}{sc A-0 (black), sc A-1 (red), sc A-2 (green), sc A-3
(blue)}\\%
{\includegraphics[
height=2.2943in,
width=2.9153in
]%
{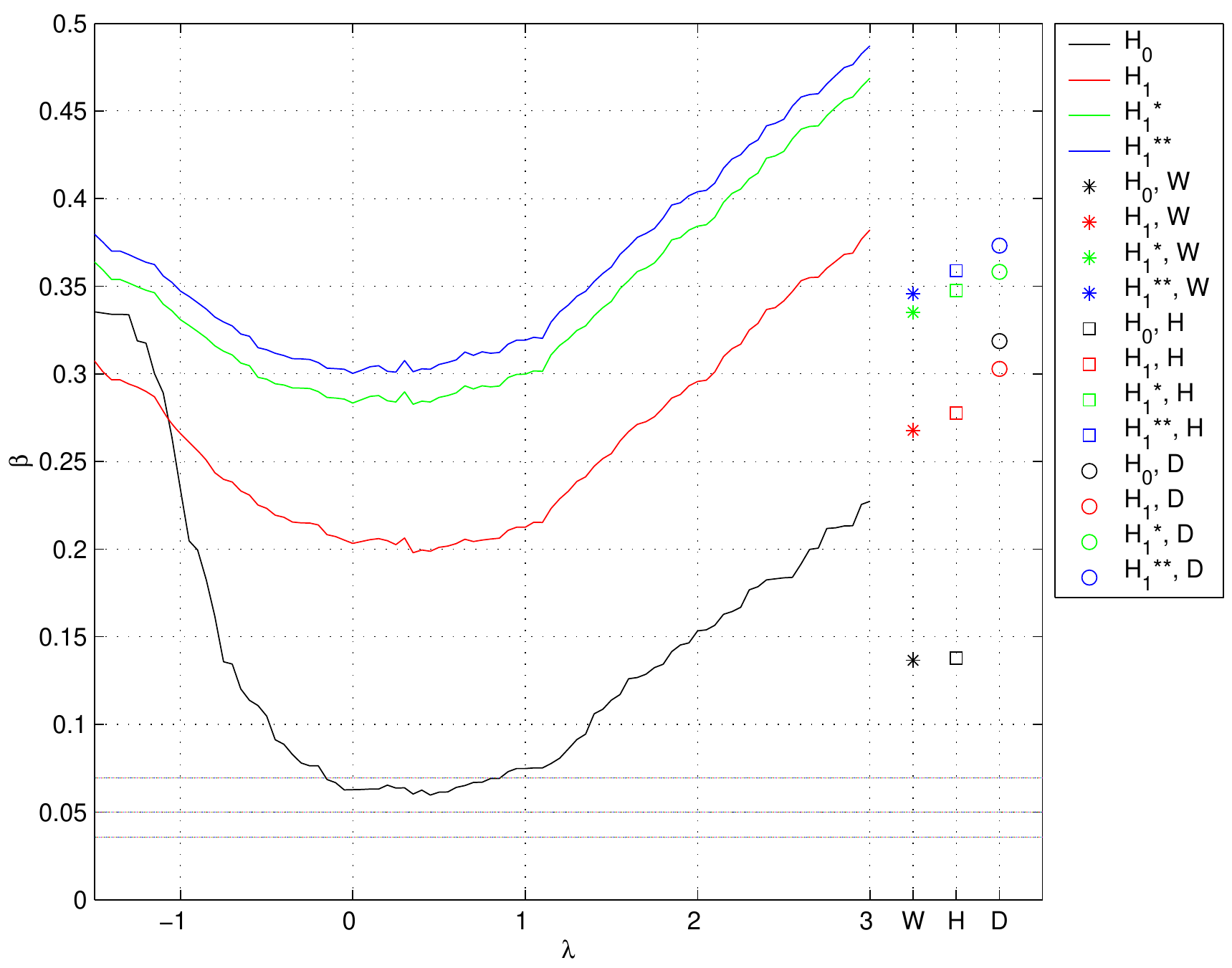}%
}
&
{\includegraphics[
height=2.2943in,
width=2.9153in
]%
{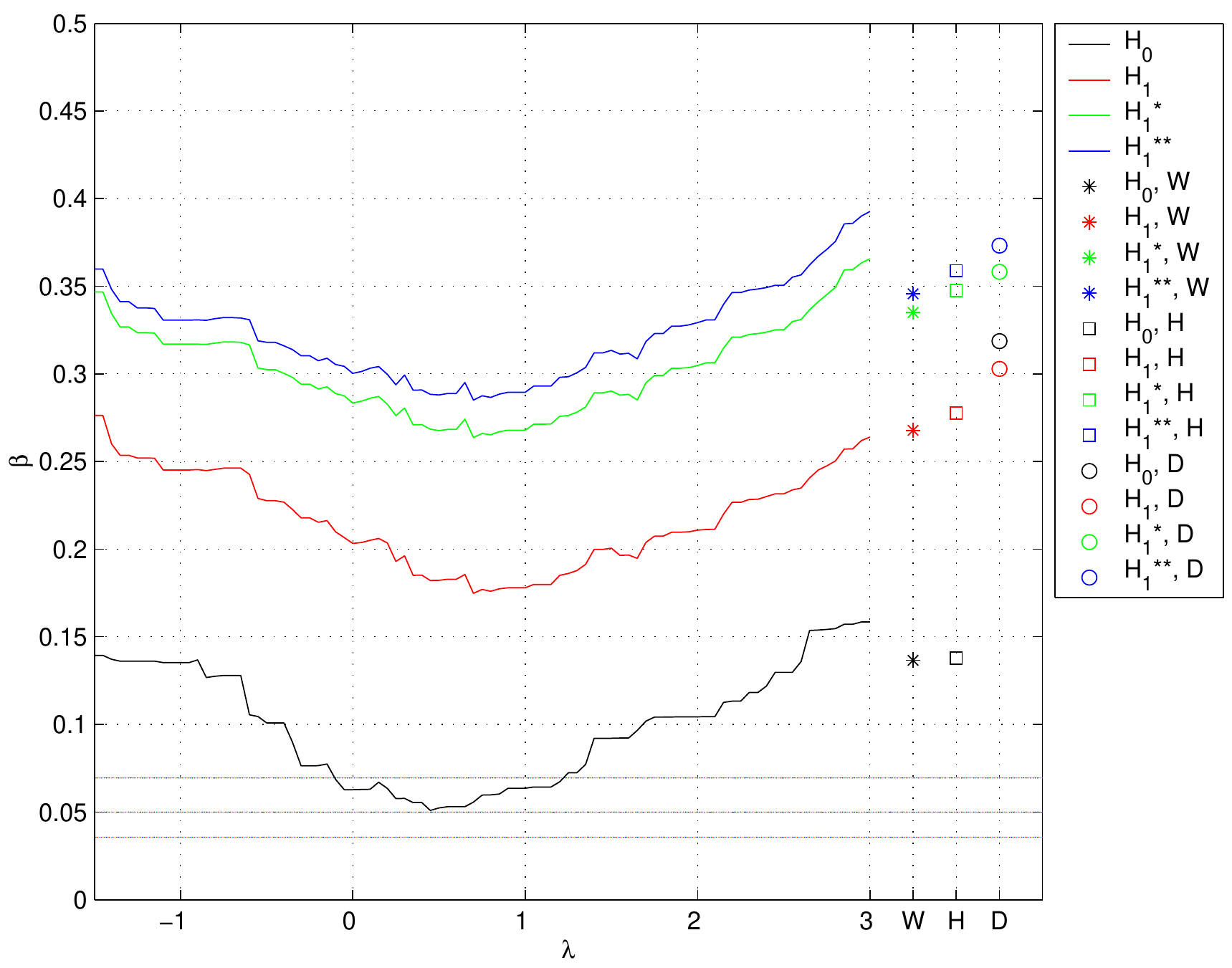}%
}
\\
\multicolumn{2}{c}{sc B-0 (black), sc B-1 (red), sc B-2 (green), sc B-3
(blue)}\\%
{\includegraphics[
height=2.2943in,
width=2.879in
]%
{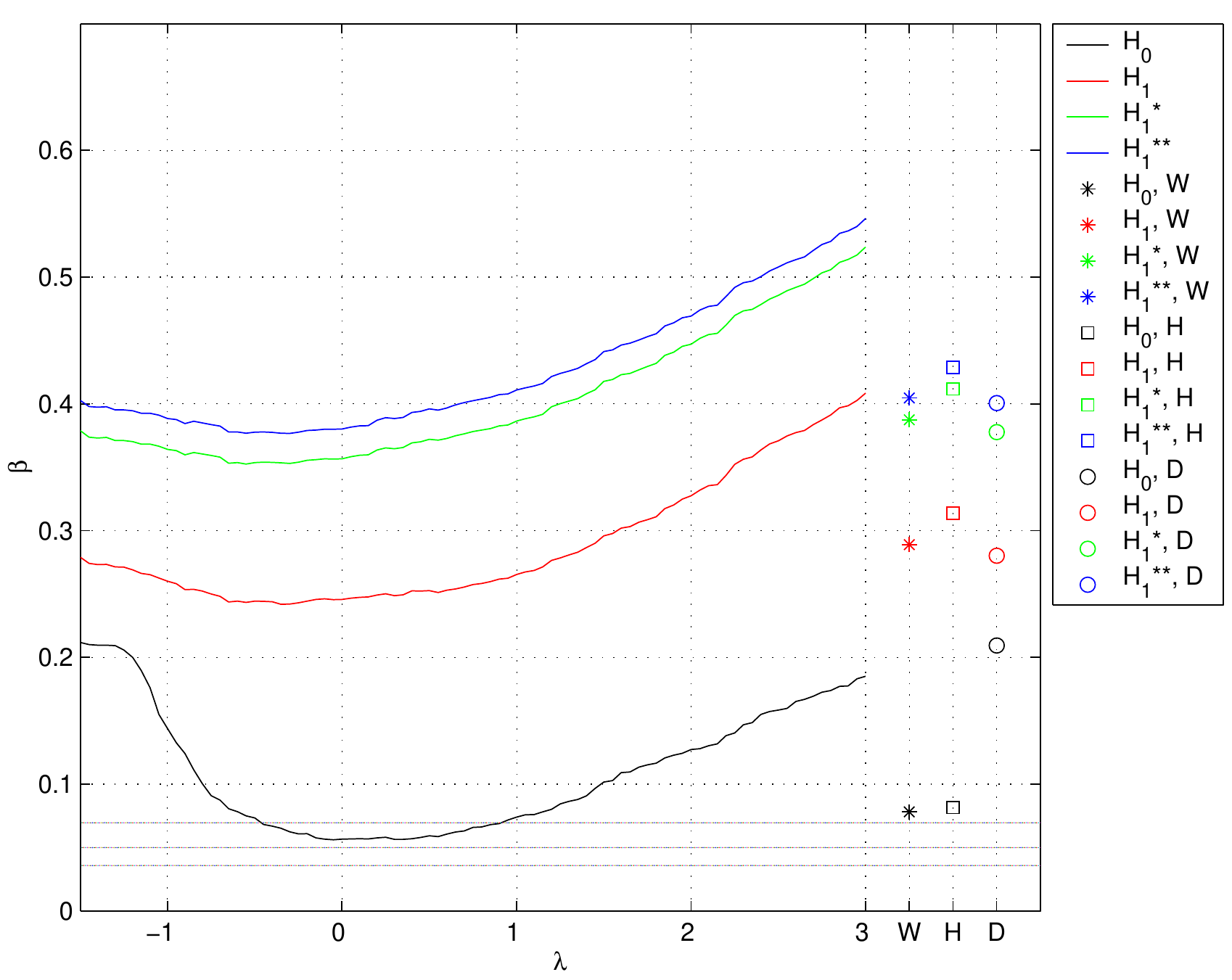}%
}
&
{\includegraphics[
height=2.2943in,
width=2.879in
]%
{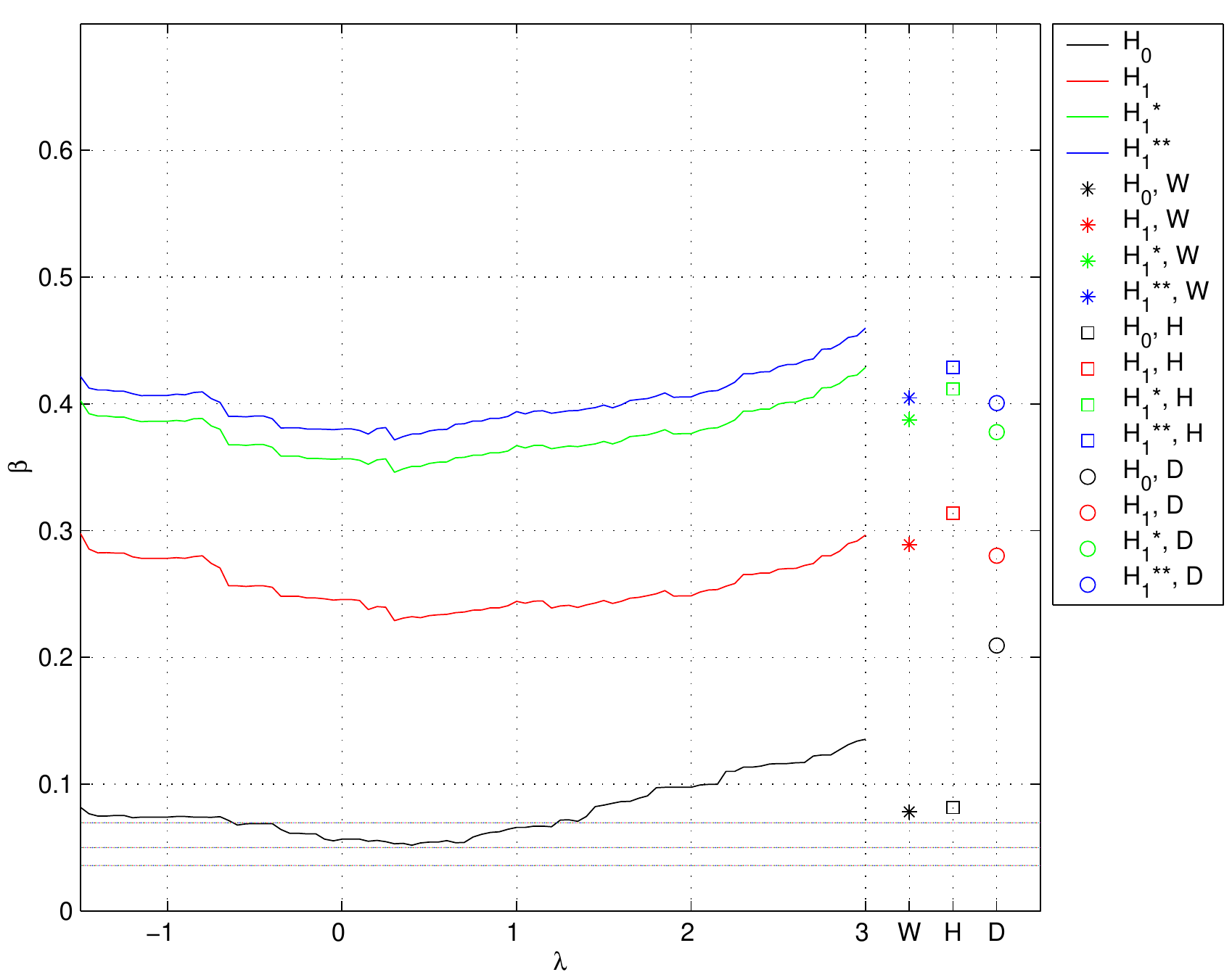}%
}
\\
\multicolumn{2}{c}{sc C-0 (black), sc C-1 (red), sc C-2 (green), sc C-3
(blue)}\\%
{\includegraphics[
height=2.2943in,
width=2.879in
]%
{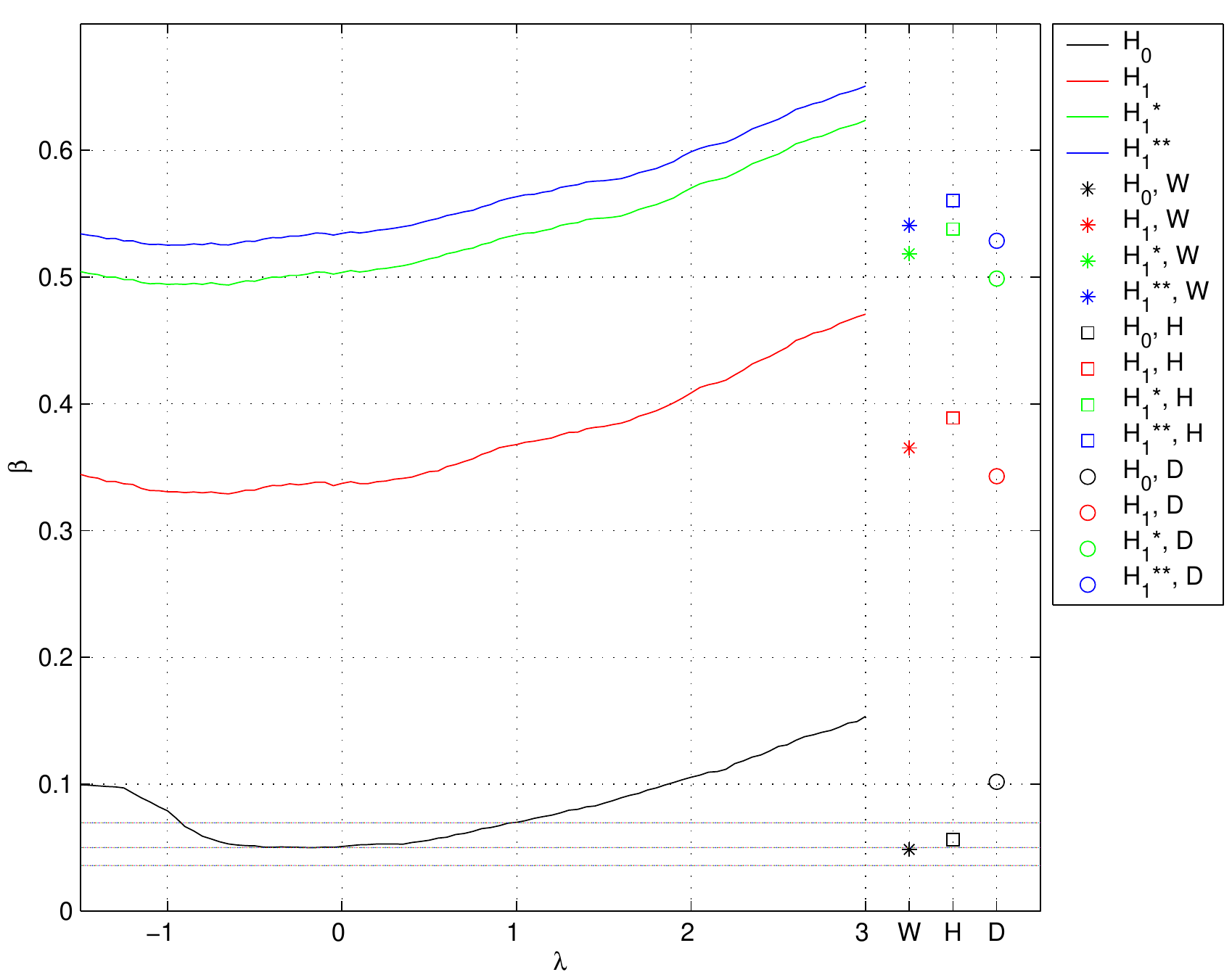}%
}
&
{\includegraphics[
height=2.2943in,
width=2.879in
]%
{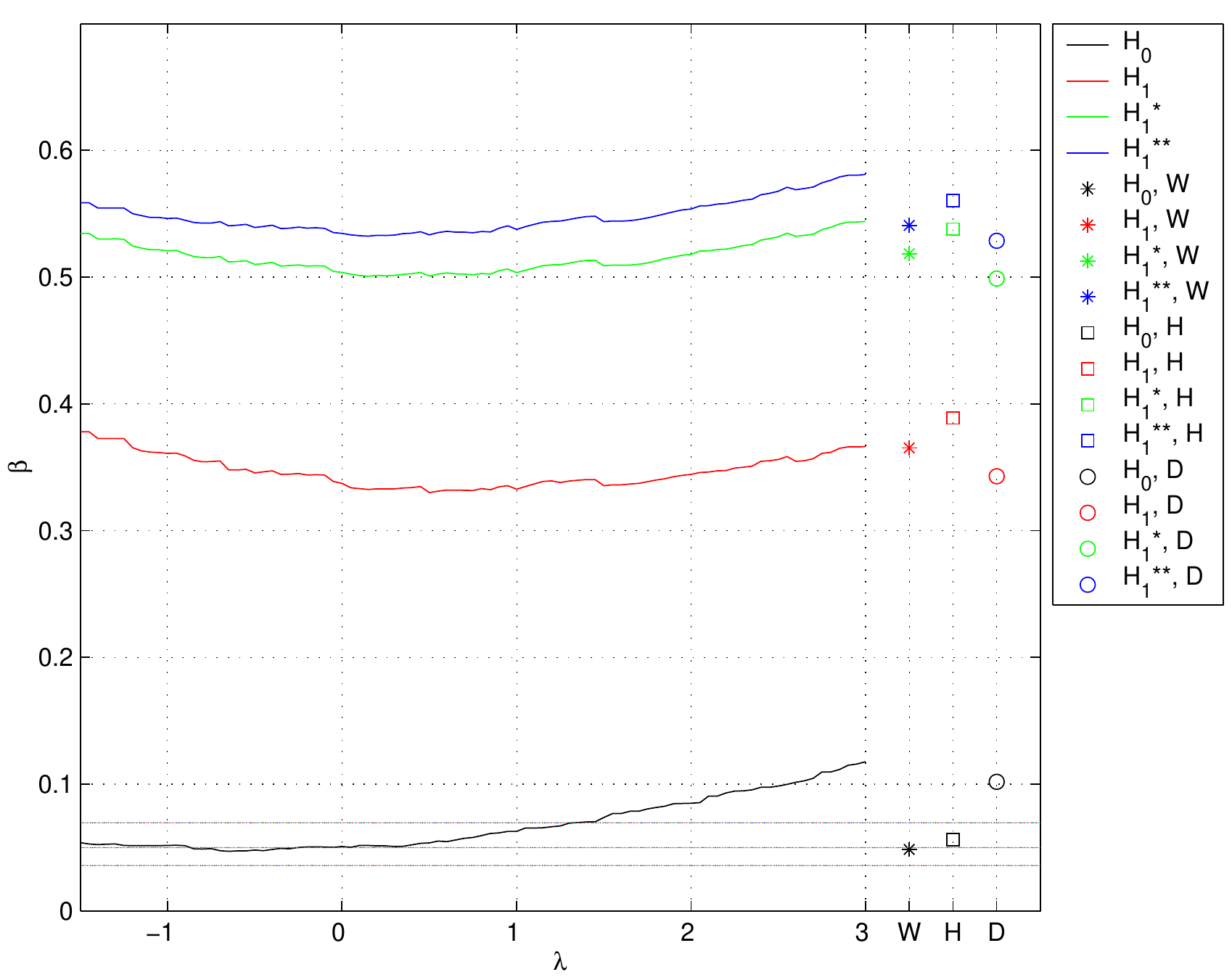}%
}
\end{tabular}
\caption{Simulated sizes (black) and powers (red, green, blue) for scenarios A,B,C (small/big proportions).\label{fig1}}%
\end{figure}%
%

\begin{figure}[htbp]  \centering
\begin{tabular}
[c]{cc}%
$\{T_{\lambda}(\overline{\boldsymbol{p}},\boldsymbol{p}%
(\widetilde{\boldsymbol{\theta}}),\boldsymbol{p}(\widehat{\boldsymbol{\theta}%
}))\}_{\lambda\in(-1.5,3)}$ & $\{S_{\lambda}(\boldsymbol{p}%
(\widetilde{\boldsymbol{\theta}}),\boldsymbol{p}(\widehat{\boldsymbol{\theta}%
}))\}_{\lambda\in(-1.5,3)}$\\
\multicolumn{2}{c}{$W(\widetilde{\boldsymbol{\theta}}%
,\widehat{\boldsymbol{\theta}})$, $H(\widetilde{\boldsymbol{\theta}%
},\widehat{\boldsymbol{\theta}})$, $D(\overline{\boldsymbol{\theta}%
},\widetilde{\boldsymbol{\theta}},\widehat{\boldsymbol{\theta}})$}\\
\multicolumn{2}{c}{sc D-0 (black), sc D-1 (red), sc D-2 (green), sc D-3
(blue)}\\%
{\includegraphics[
height=2.2943in,
width=2.9144in
]%
{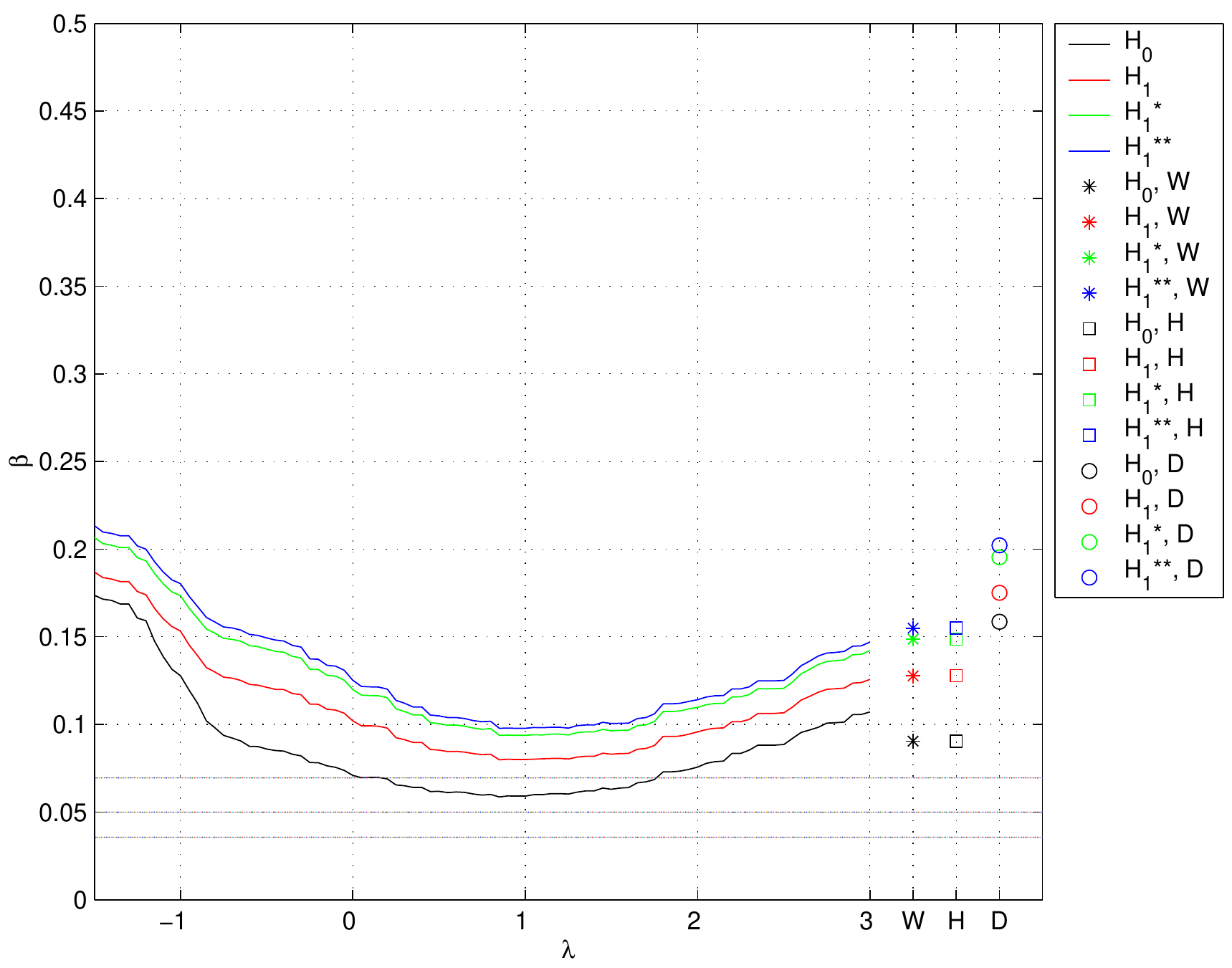}%
}
&
{\includegraphics[
height=2.2943in,
width=2.9144in
]%
{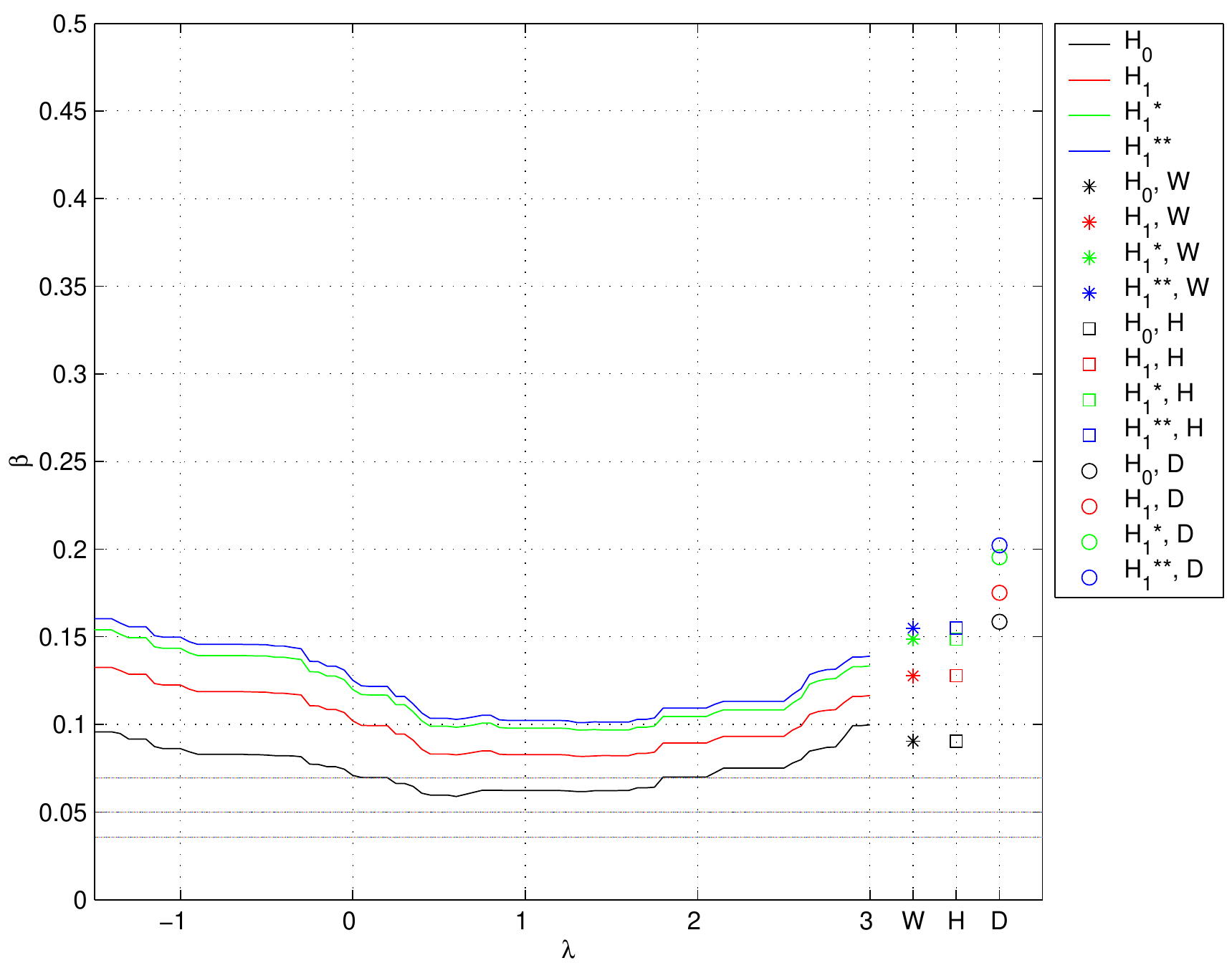}%
}
\\
\multicolumn{2}{c}{sc E-0 (black), sc E-1 (red), sc E-2 (green), sc E-3
(blue)}\\%
{\includegraphics[
height=2.2943in,
width=2.9144in
]%
{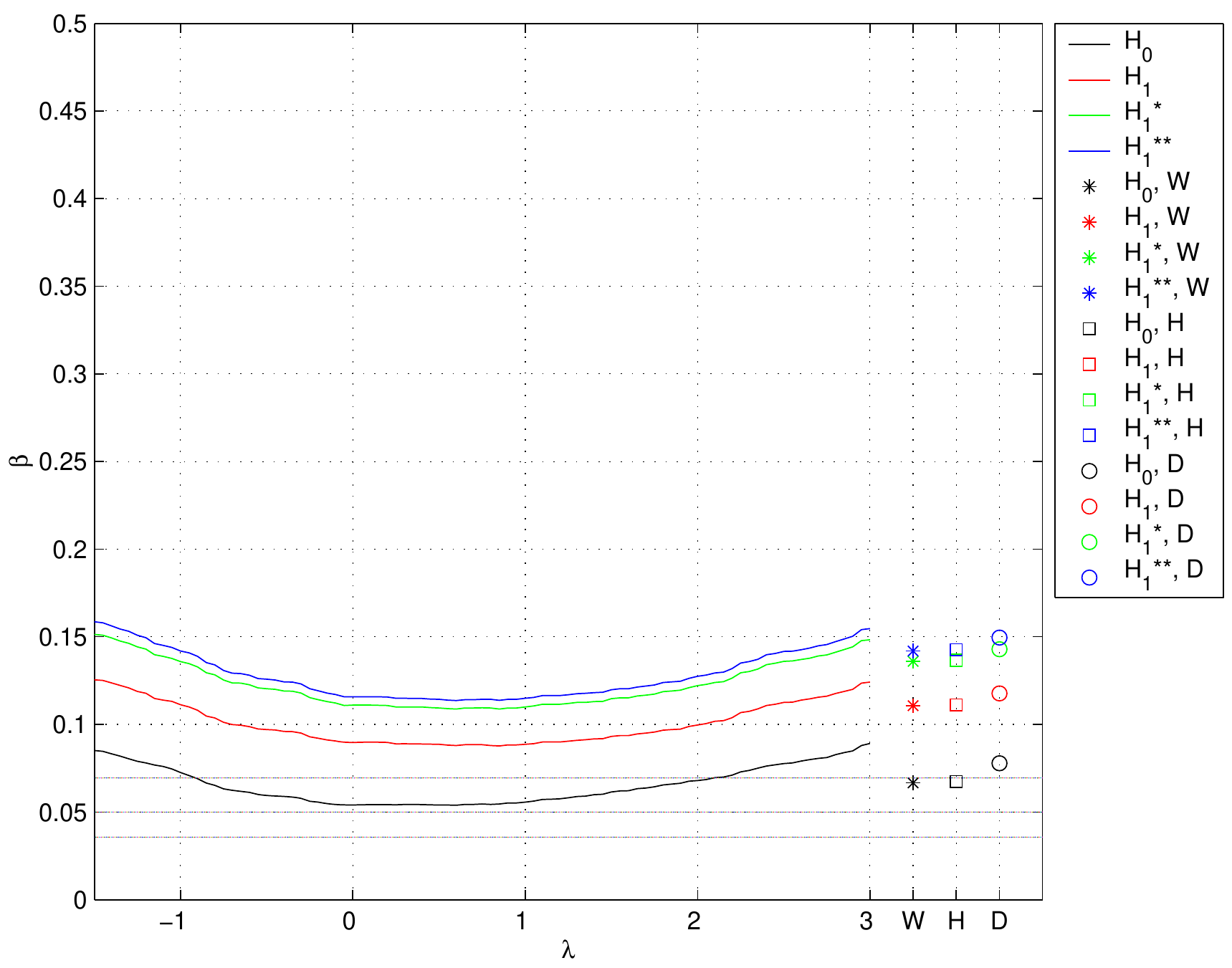}%
}
&
{\includegraphics[
height=2.2943in,
width=2.9144in
]%
{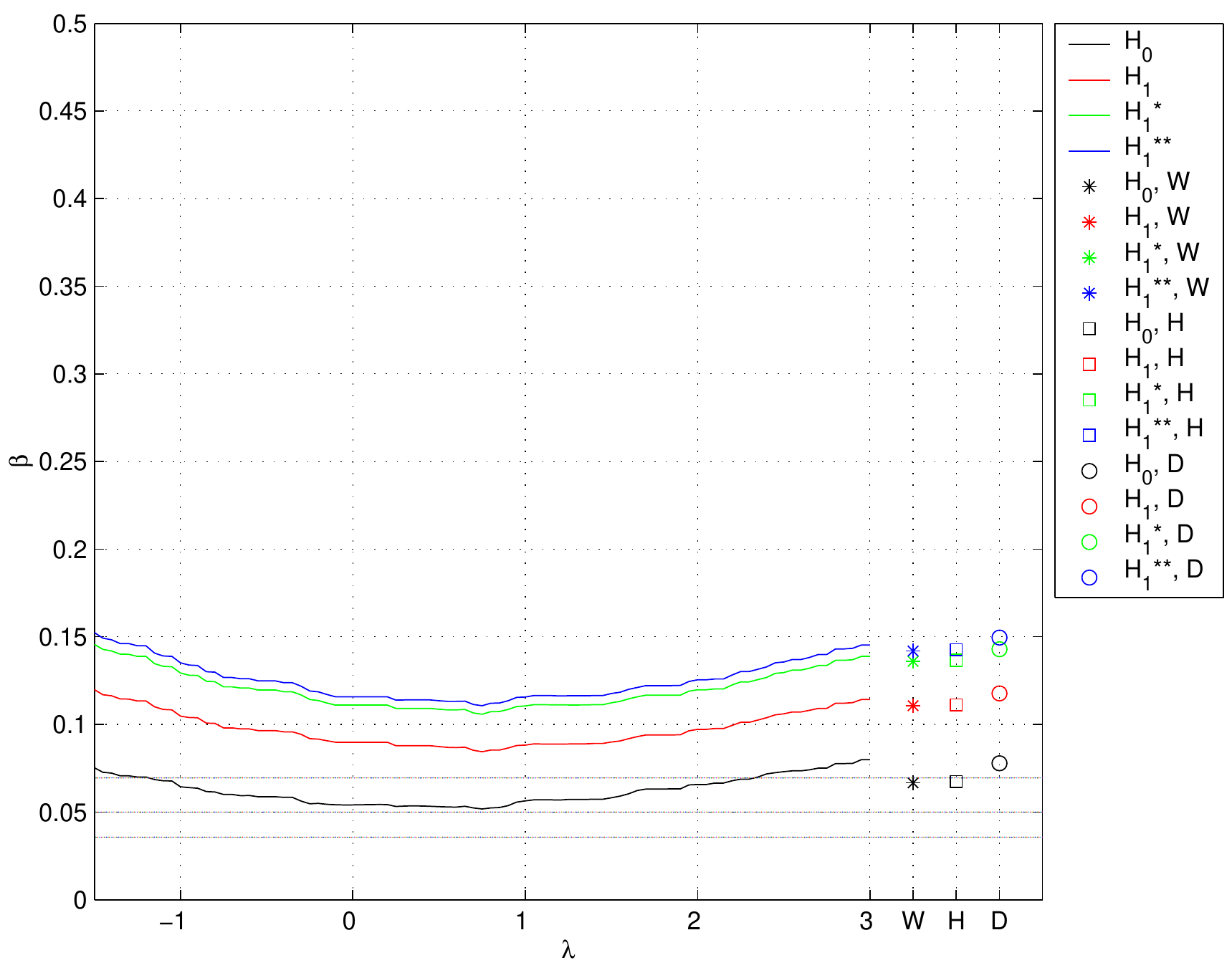}%
}
\\
\multicolumn{2}{c}{sc F-0 (black), sc F-1 (red), sc F-2 (green), sc F-3
(blue)}\\%
{\includegraphics[
height=2.2943in,
width=2.9144in
]%
{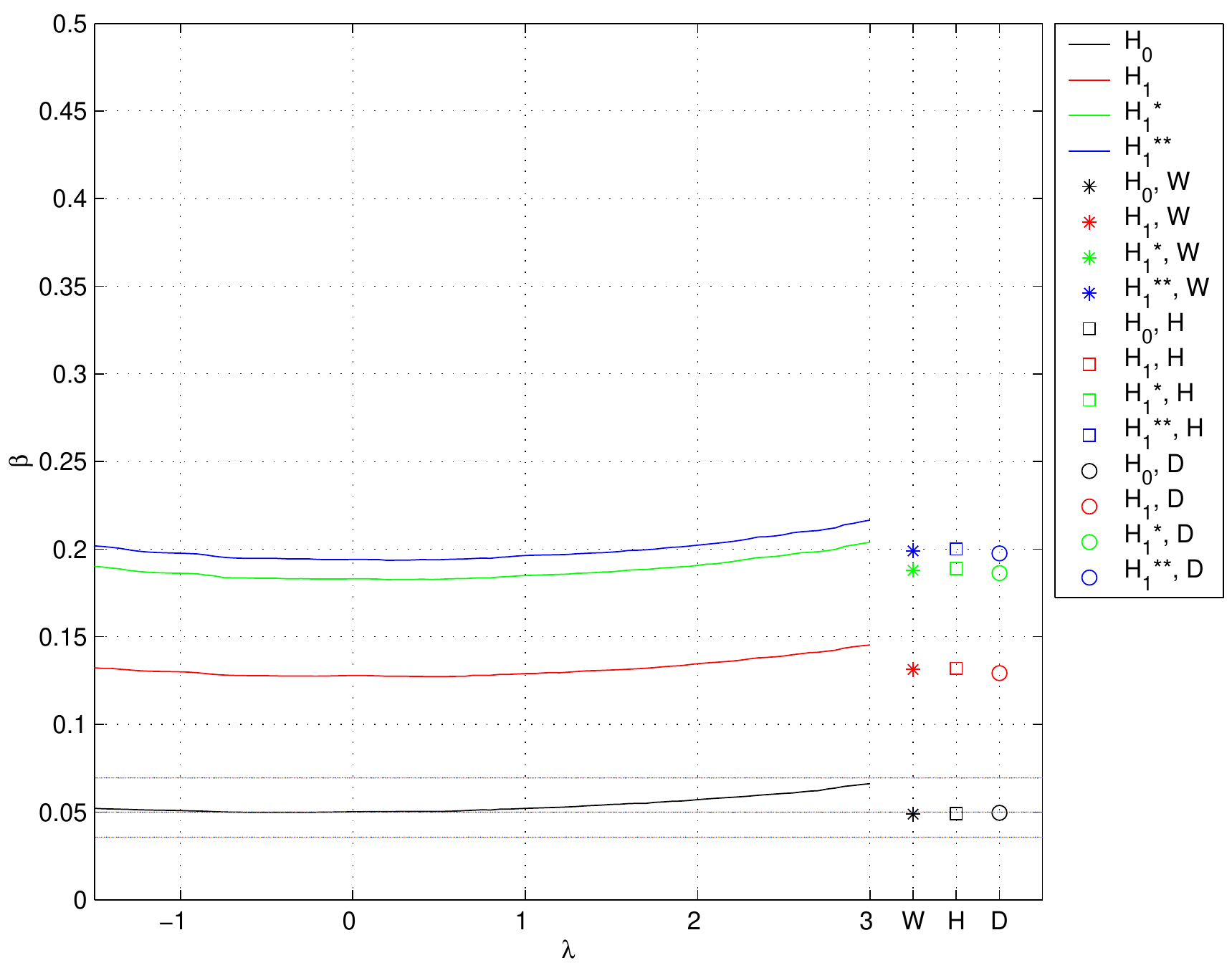}%
}
&
{\includegraphics[
height=2.2943in,
width=2.9144in
]%
{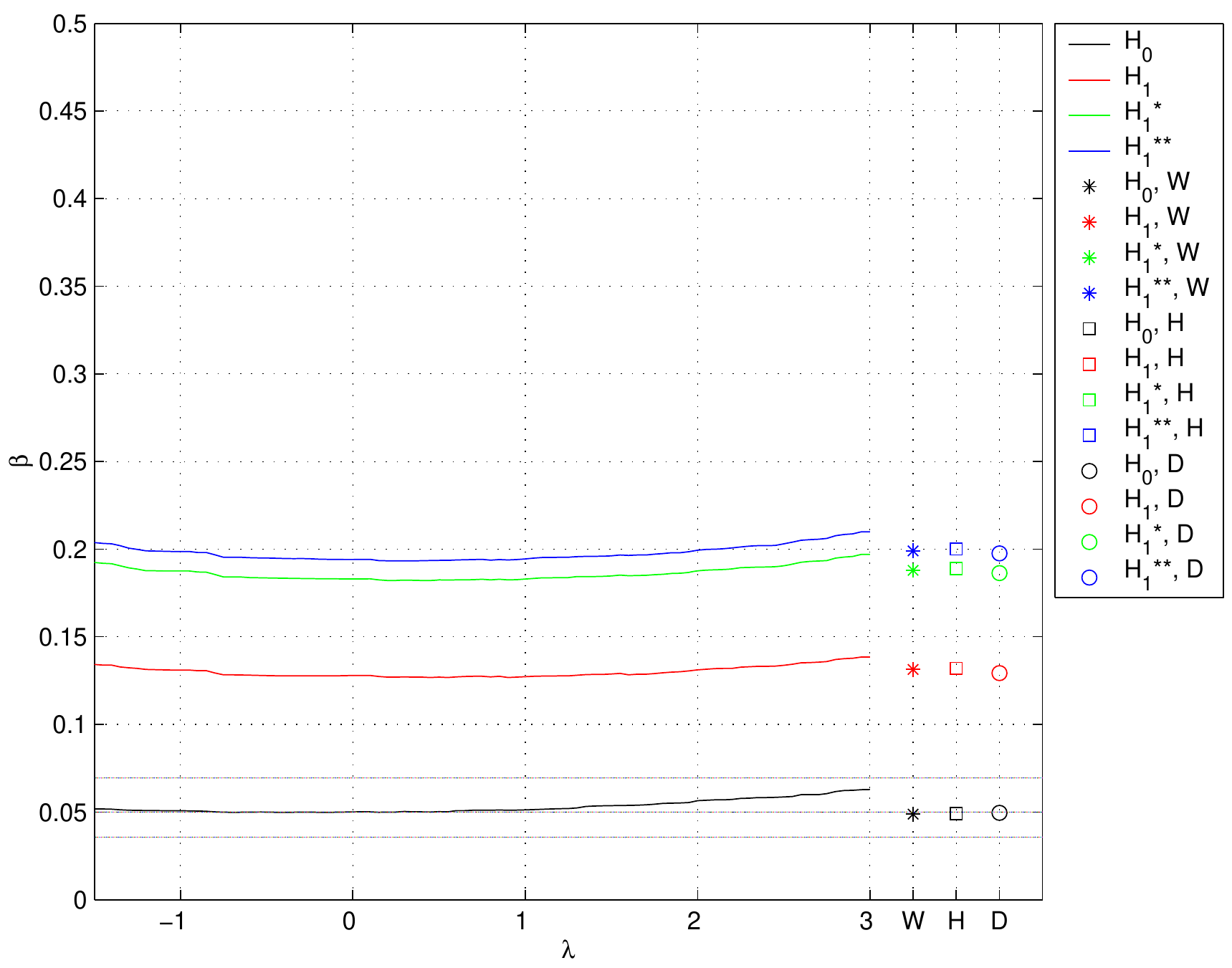}%
}
\end{tabular}
\caption{Simulated sizes (black) and powers (red, green, blue) for scenarios D,E,F (intermediate proportions).\label{fig2}}%
\end{figure}%
%

\begin{figure}[htbp]  \centering
\begin{tabular}
[c]{cc}%
$\{T_{\lambda}(\overline{\boldsymbol{p}},\boldsymbol{p}%
(\widetilde{\boldsymbol{\theta}}),\boldsymbol{p}(\widehat{\boldsymbol{\theta}%
}))\}_{\lambda\in(-1.5,3)}$ & $\{S_{\lambda}(\boldsymbol{p}%
(\widetilde{\boldsymbol{\theta}}),\boldsymbol{p}(\widehat{\boldsymbol{\theta}%
}))\}_{\lambda\in(-1.5,3)}$\\
\multicolumn{2}{c}{$W(\widetilde{\boldsymbol{\theta}}%
,\widehat{\boldsymbol{\theta}})$, $H(\widetilde{\boldsymbol{\theta}%
},\widehat{\boldsymbol{\theta}})$, $D(\overline{\boldsymbol{\theta}%
},\widetilde{\boldsymbol{\theta}},\widehat{\boldsymbol{\theta}})$}\\
\multicolumn{2}{c}{sc A-0 (black), sc A-1 (red), sc A-2 (green), sc A-3
(blue)}\\%
{\includegraphics[
height=2.2943in,
width=2.9144in
]%
{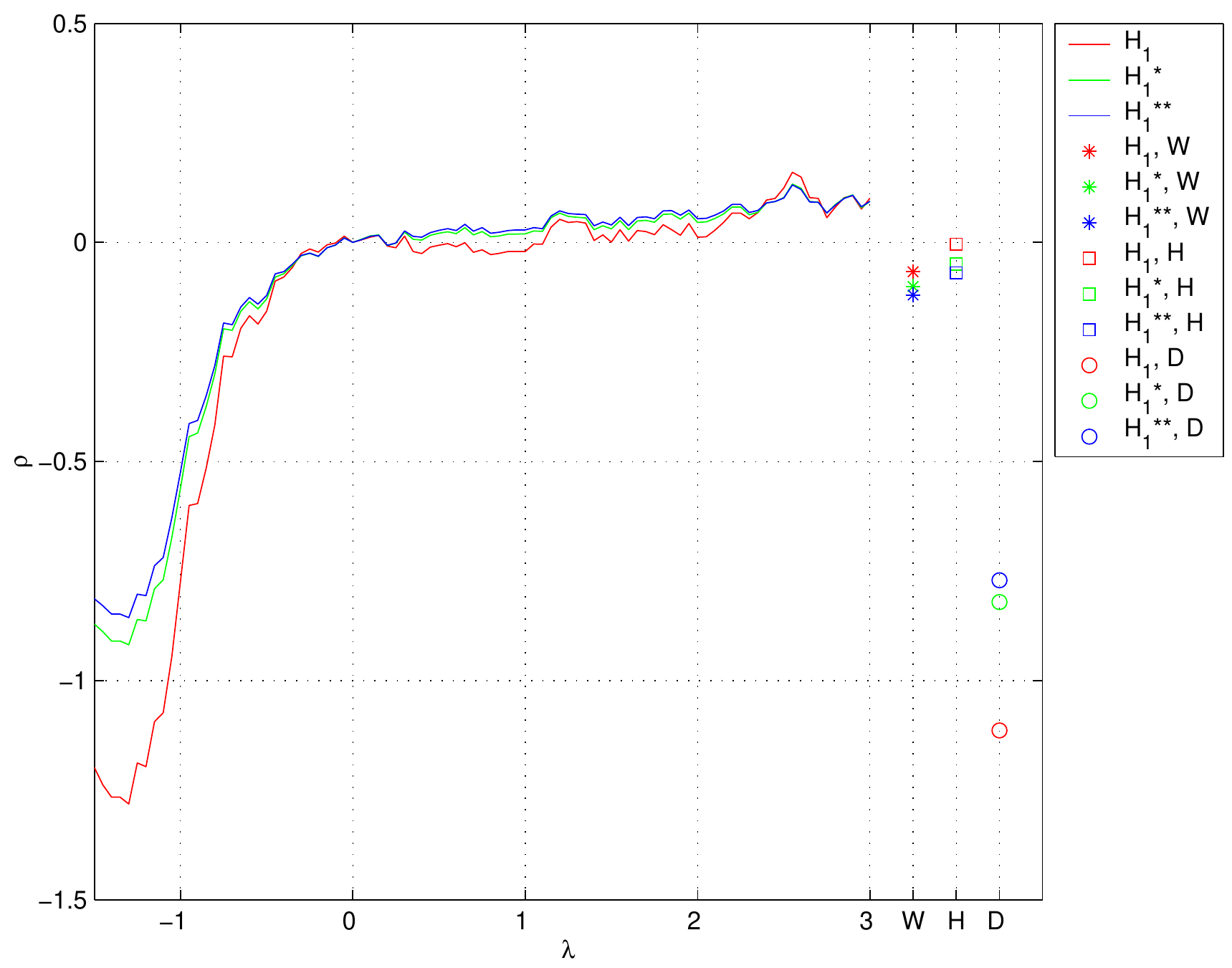}%
}
&
{\includegraphics[
height=2.2943in,
width=2.9144in
]%
{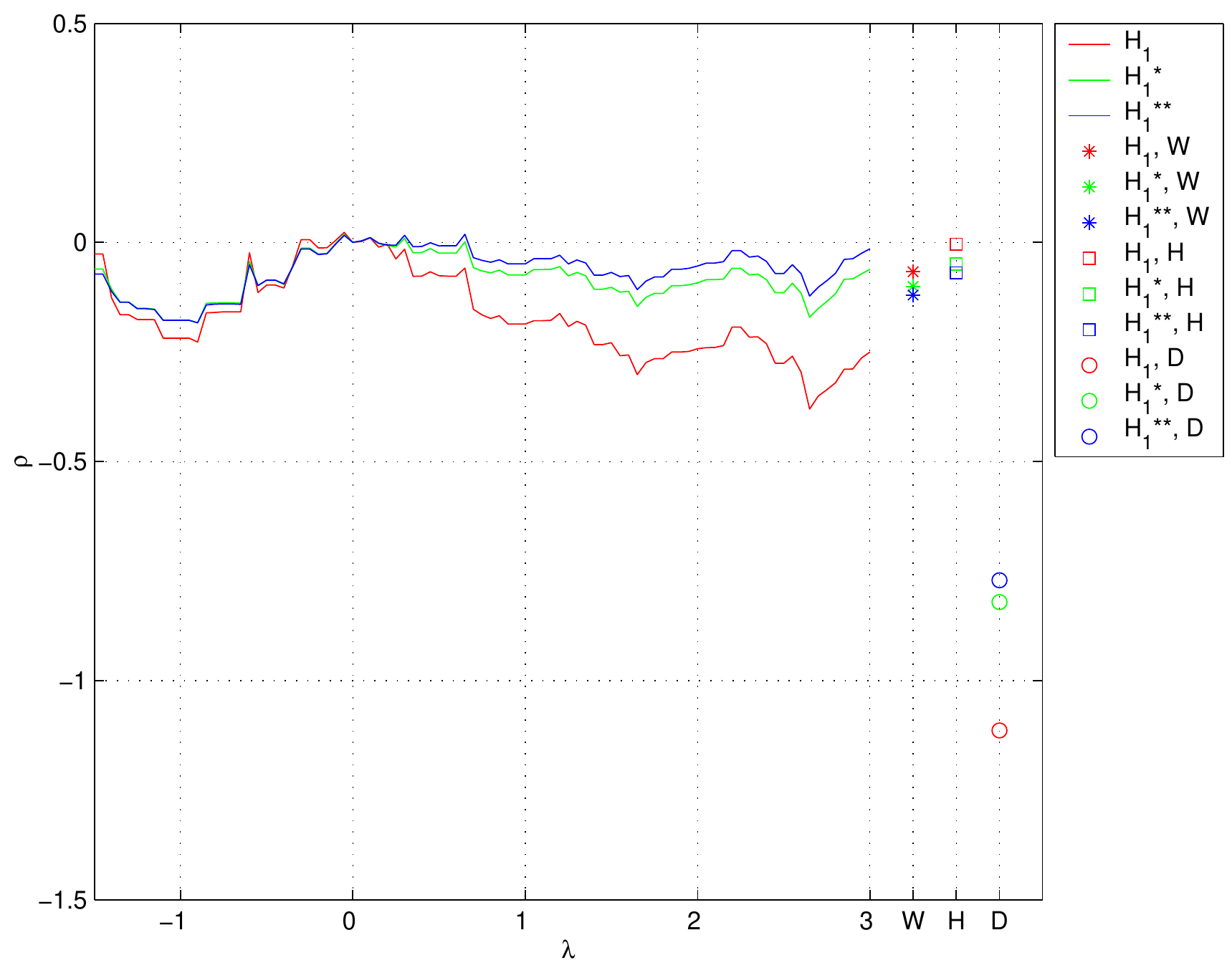}%
}
\\
\multicolumn{2}{c}{sc B-0 (black), sc B-1 (red), sc B-2 (green), sc B-3
(blue)}\\%
{\includegraphics[
height=2.2943in,
width=2.9144in
]%
{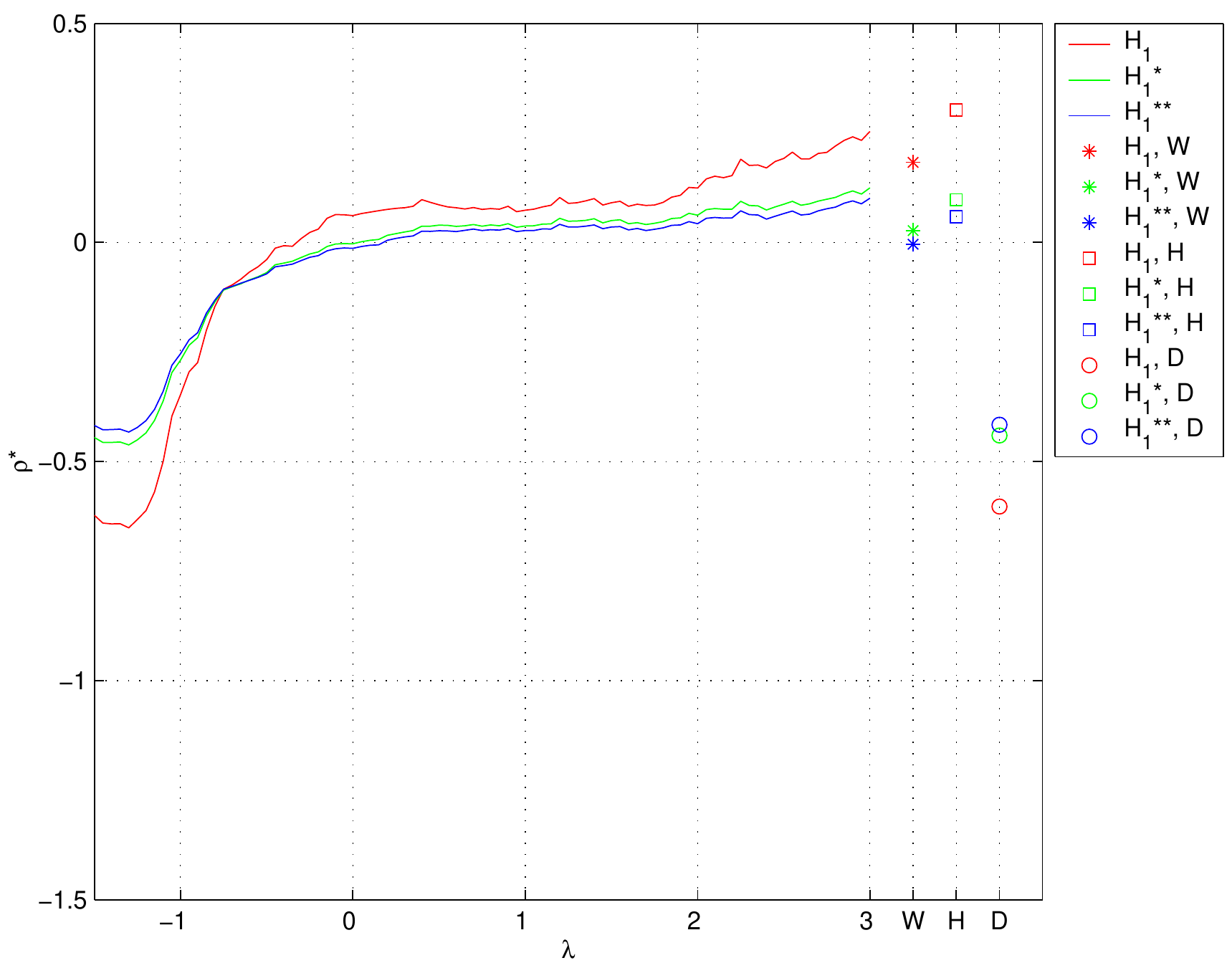}%
}
&
{\includegraphics[
height=2.2943in,
width=2.9144in
]%
{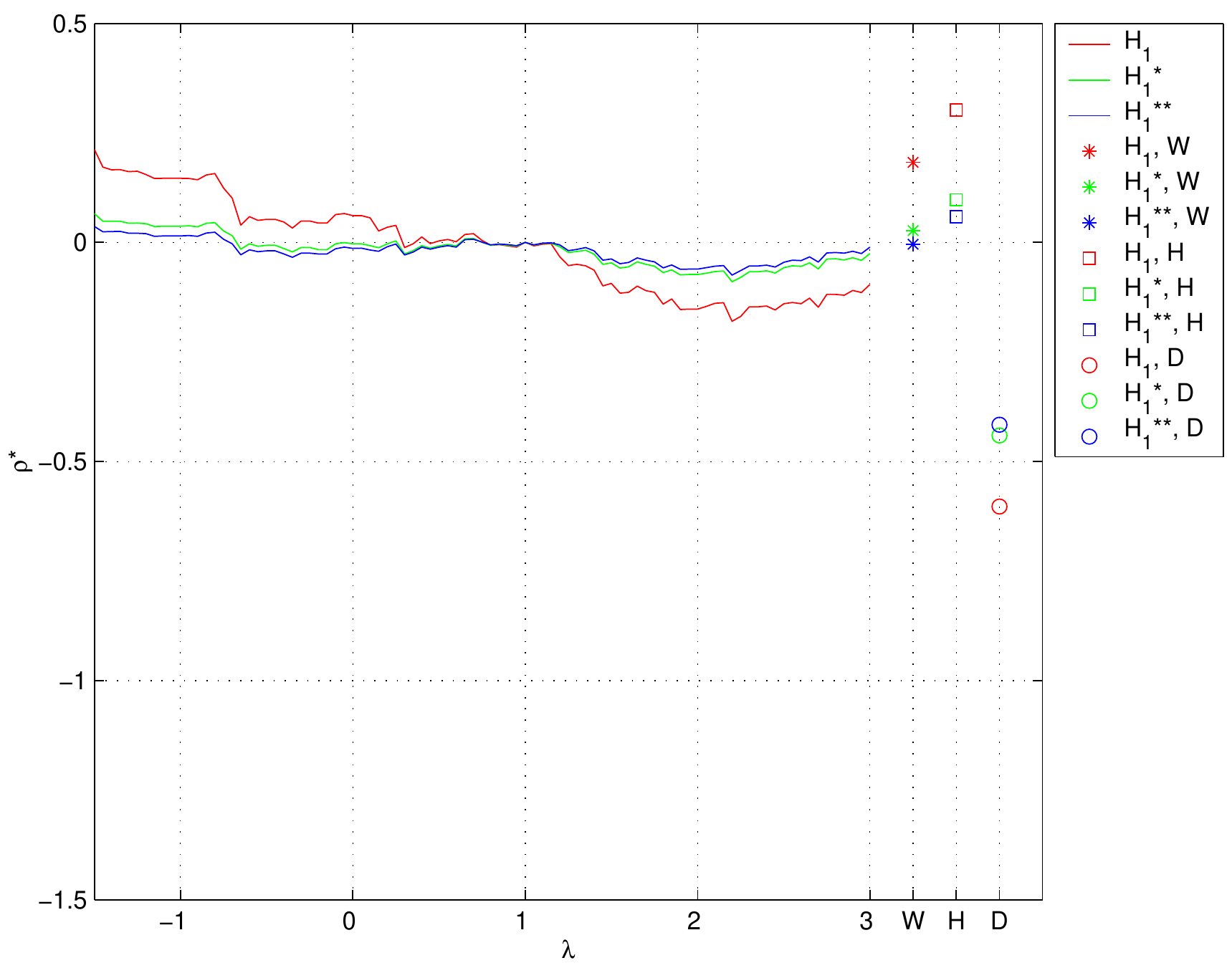}%
}
\\
\multicolumn{2}{c}{sc C-0 (black), sc C-1 (red), sc C-2 (green), sc C-3
(blue)}\\%
{\includegraphics[
height=2.2943in,
width=2.9144in
]%
{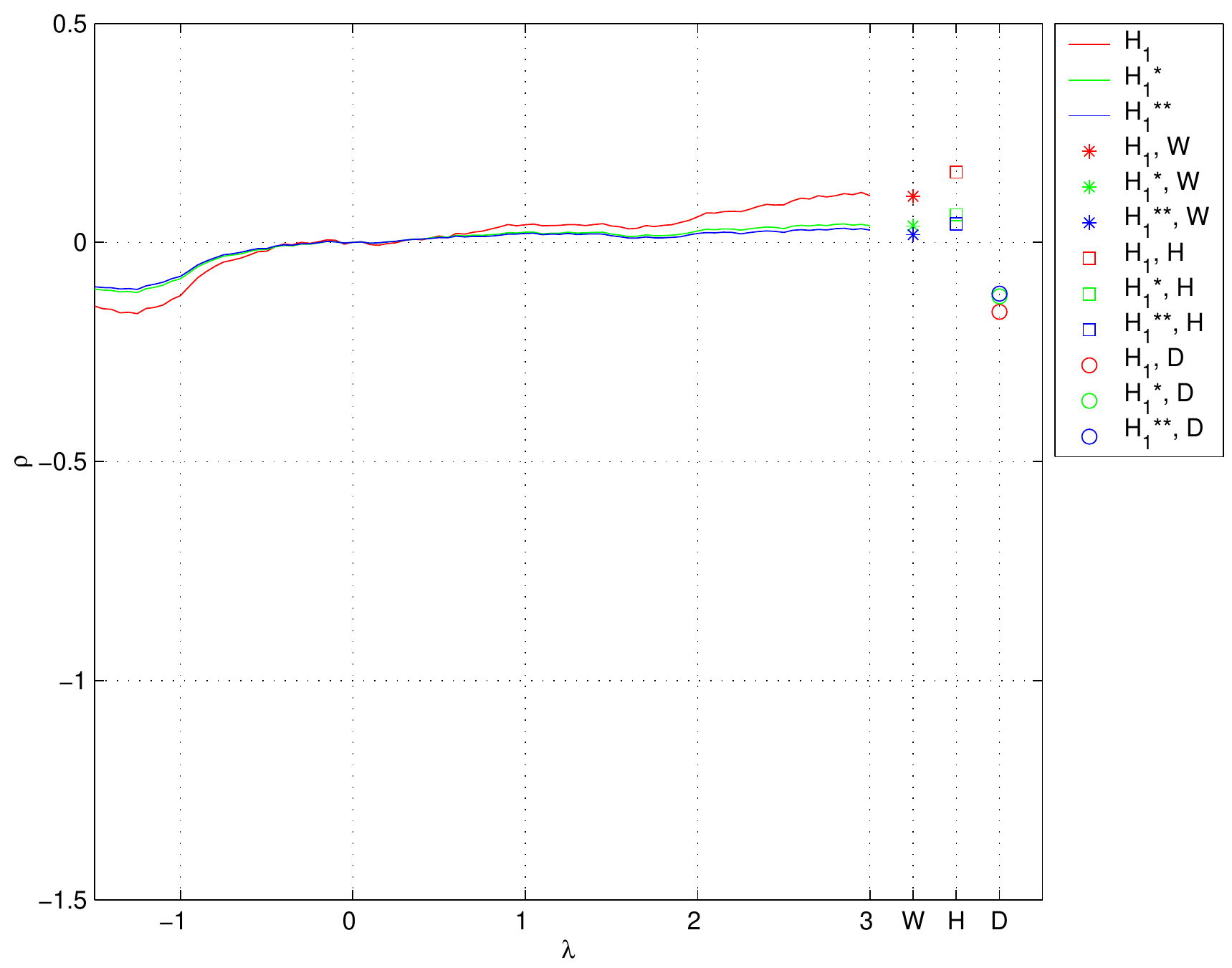}%
}
&
{\includegraphics[
height=2.2943in,
width=2.9144in
]%
{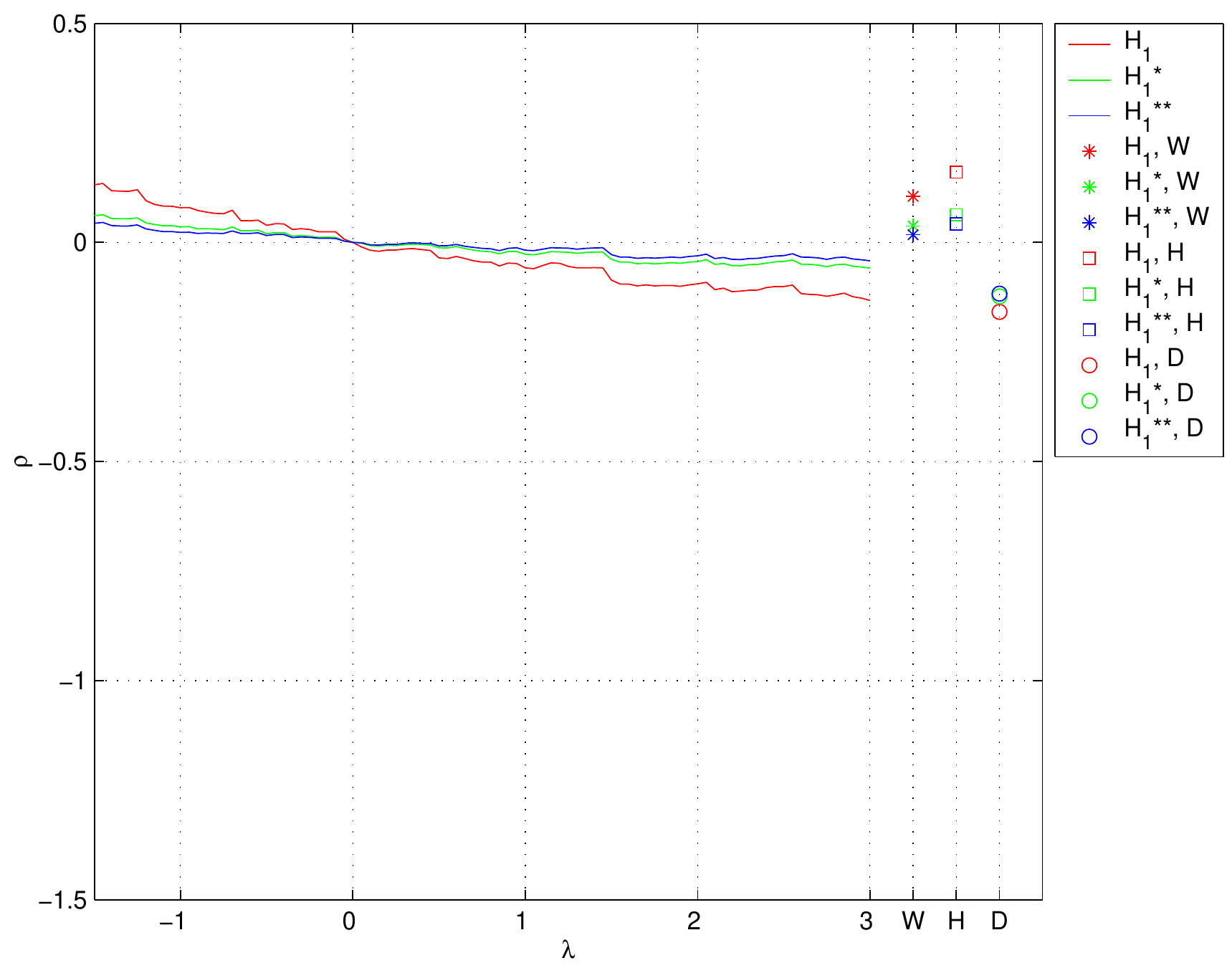}%
}
\end{tabular}
\caption{Efficiencies for scenarios A,B,C (small/big proportions).\label{fig3}}%
\end{figure}%
%

\begin{figure}[htbp]  \centering
\begin{tabular}
[c]{cc}%
$\{T_{\lambda}(\overline{\boldsymbol{p}},\boldsymbol{p}%
(\widetilde{\boldsymbol{\theta}}),\boldsymbol{p}(\widehat{\boldsymbol{\theta}%
}))\}_{\lambda\in(-1.5,3)}$ & $\{S_{\lambda}(\boldsymbol{p}%
(\widetilde{\boldsymbol{\theta}}),\boldsymbol{p}(\widehat{\boldsymbol{\theta}%
}))\}_{\lambda\in(-1.5,3)}$\\
\multicolumn{2}{c}{$W(\widetilde{\boldsymbol{\theta}}%
,\widehat{\boldsymbol{\theta}})$, $H(\widetilde{\boldsymbol{\theta}%
},\widehat{\boldsymbol{\theta}})$, $D(\overline{\boldsymbol{\theta}%
},\widetilde{\boldsymbol{\theta}},\widehat{\boldsymbol{\theta}})$}\\
\multicolumn{2}{c}{sc C-0 (black), sc C-1 (red), sc C-2 (green), sc C-3
(blue)}\\%
\raisebox{-0.0104in}{\includegraphics[
height=2.2943in,
width=2.9144in
]%
{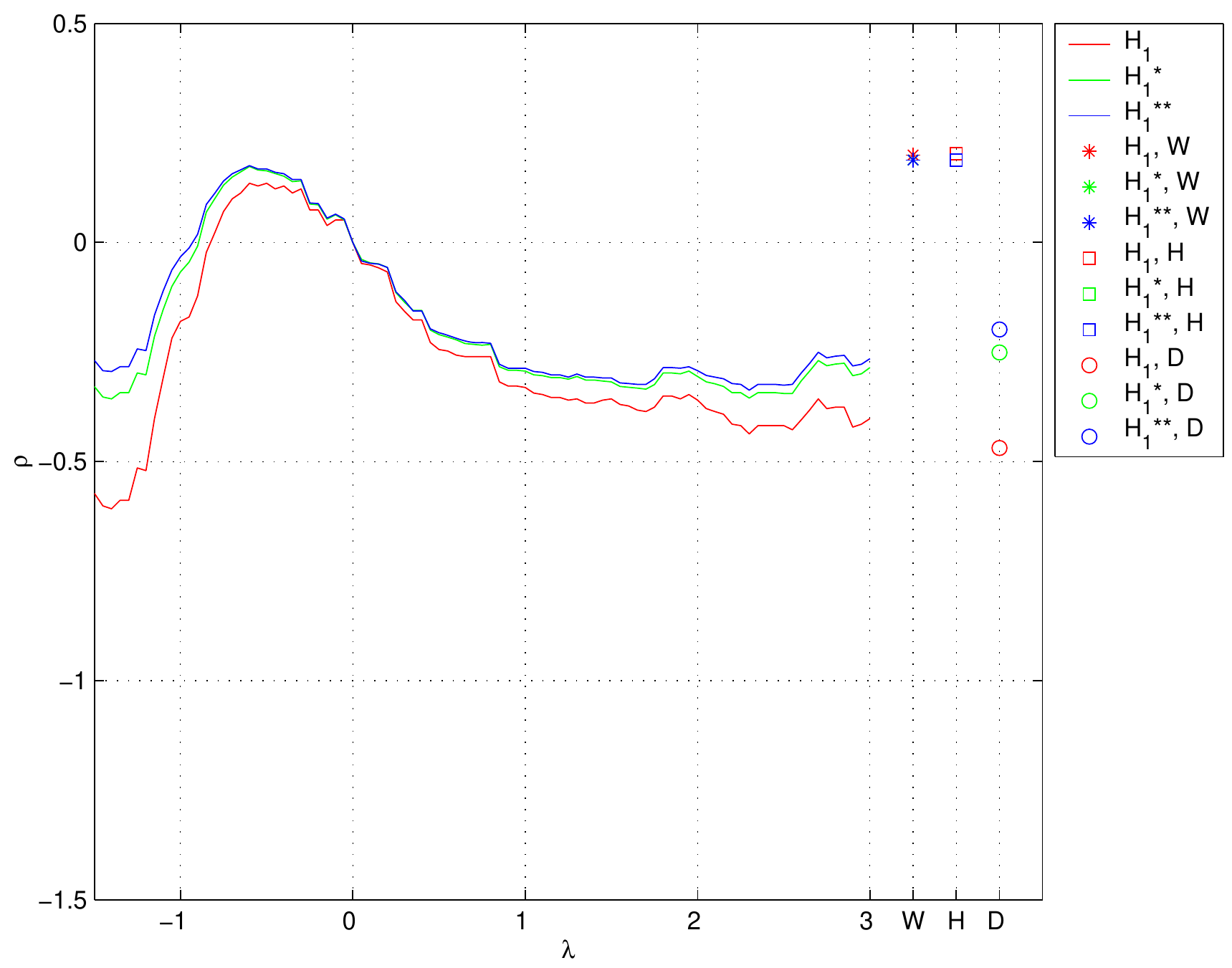}%
}
&
{\includegraphics[
height=2.2943in,
width=2.9144in
]%
{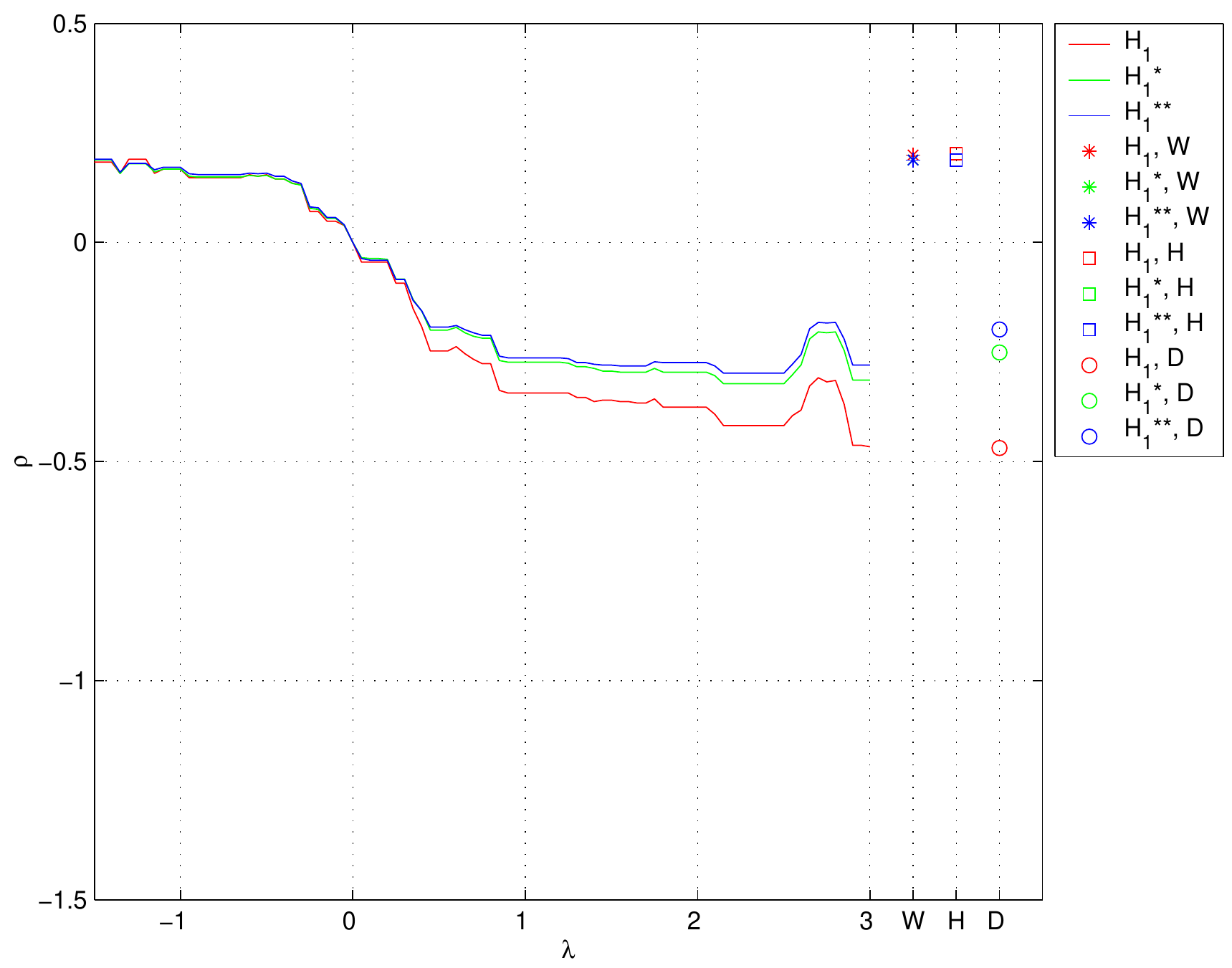}%
}
\\
\multicolumn{2}{c}{sc D-0 (black), sc D-1 (red), sc D-2 (green), sc D-3
(blue)}\\%
{\includegraphics[
height=2.2943in,
width=2.9144in
]%
{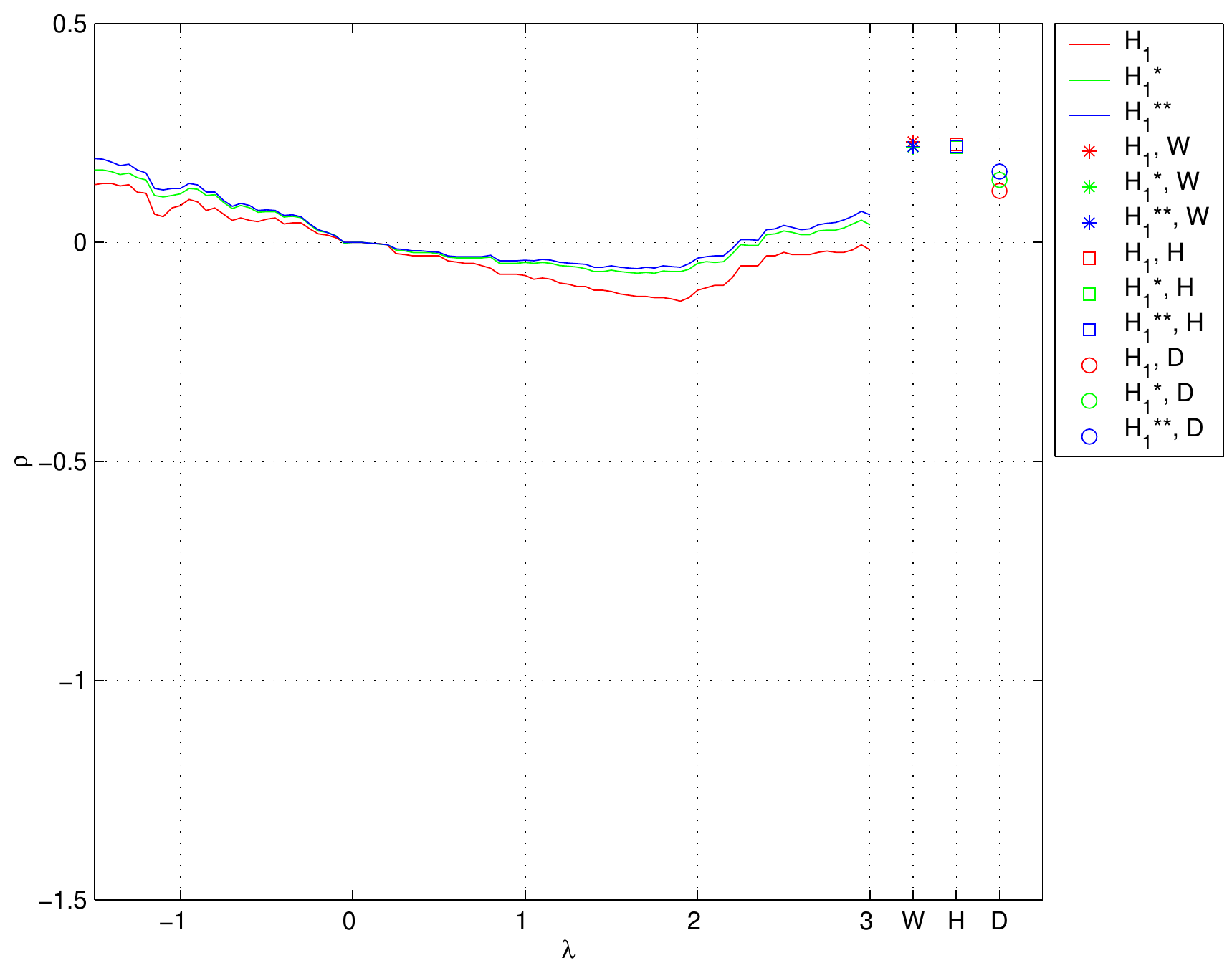}%
}
&
{\includegraphics[
height=2.2943in,
width=2.9144in
]%
{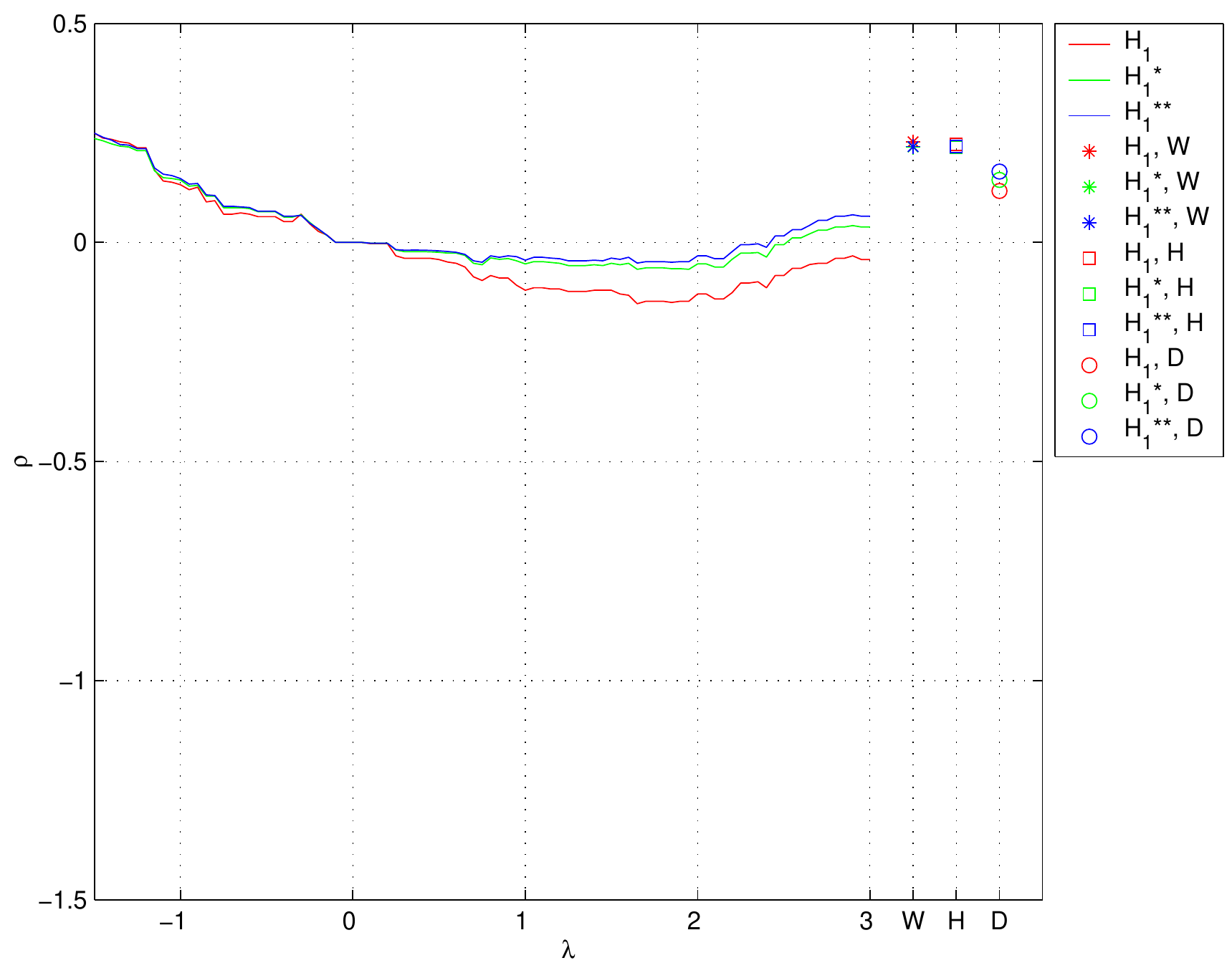}%
}
\\
\multicolumn{2}{c}{sc E-0 (black), sc E-1 (red), sc E-2 (green), sc E-3
(blue)}\\%
{\includegraphics[
height=2.2943in,
width=2.9144in
]%
{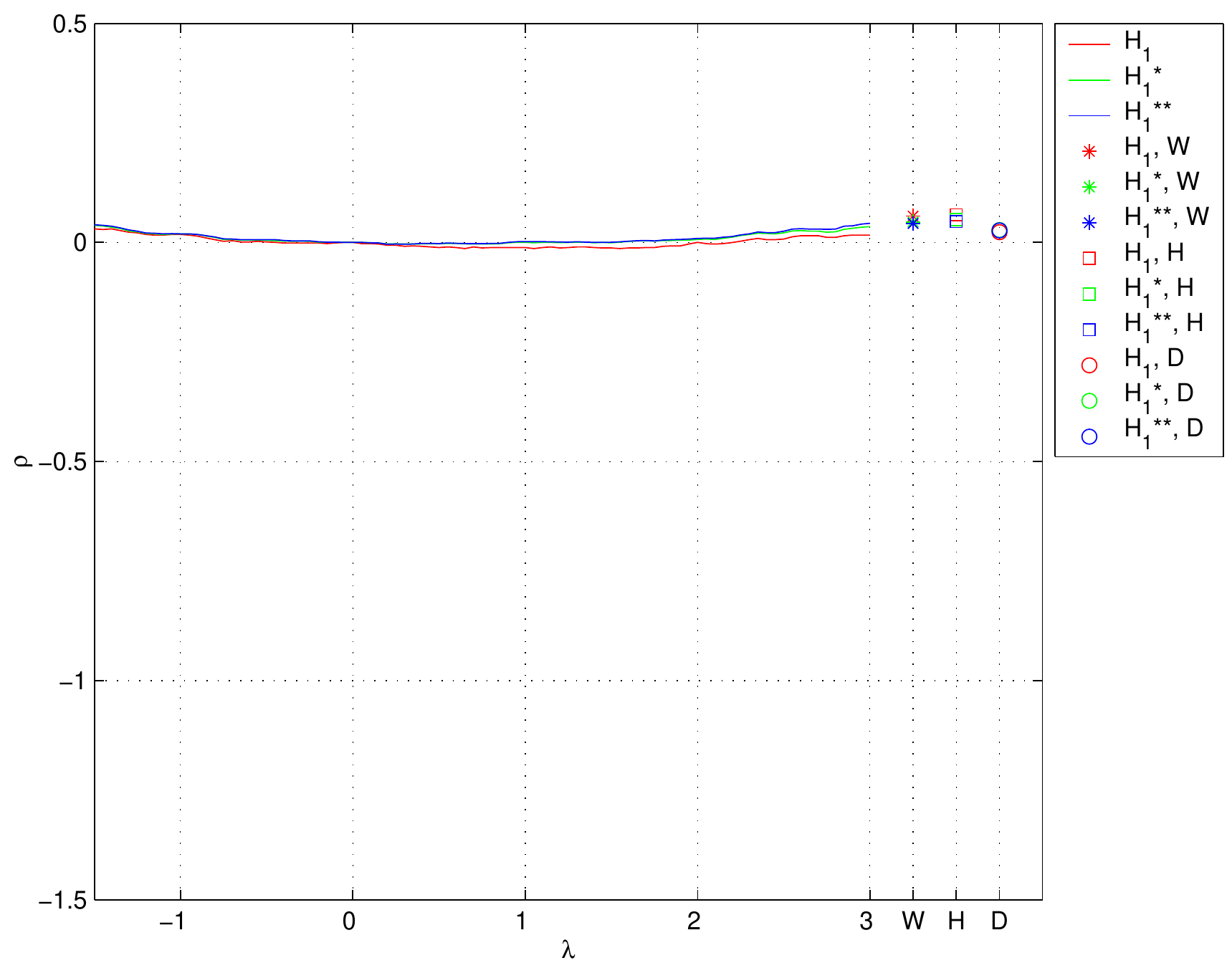}%
}
&
{\includegraphics[
height=2.2943in,
width=2.9144in
]%
{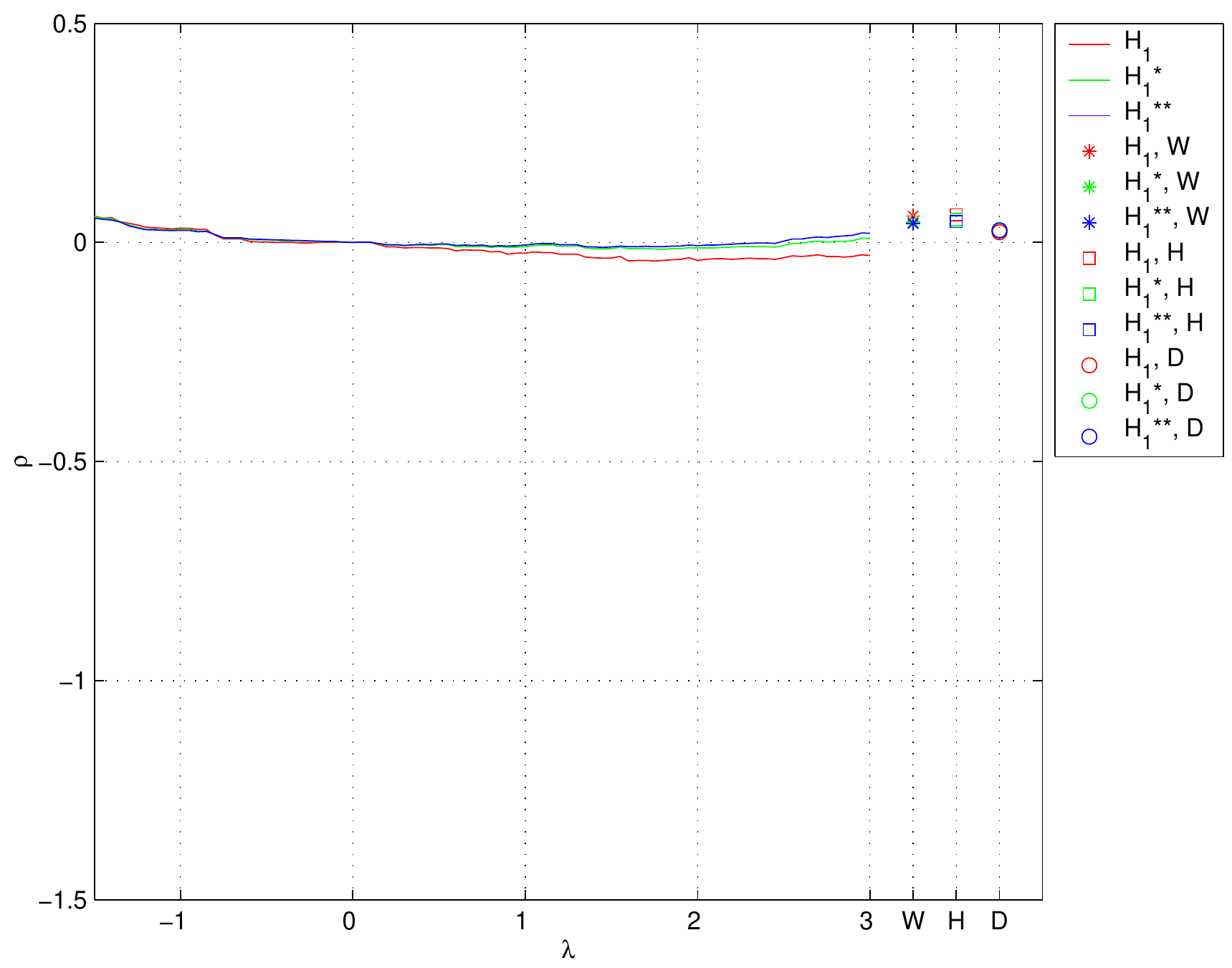}%
}
\end{tabular}
\caption{Efficiencies for scenarios C,D,E (intermediate proportions).\label{fig4}}%
\end{figure}%

In view of the plots, it is possible to propose test-statistics with better
performance in comparison with $G^{2}$ and $X^{2}$. From Figures \ref{fig1}
and \ref{fig3}, the so-called Cressie-Read test-statistic, $T_{\frac{2}{3}%
}(\overline{\boldsymbol{p}},\boldsymbol{p}(\widetilde{\boldsymbol{\theta}%
}),\boldsymbol{p}(\widehat{\boldsymbol{\theta}}))$, can be recommended for
small/big proportions either for small or moderate sample sizes. On the other
hand, $W(\widetilde{\boldsymbol{\theta}},\widehat{\boldsymbol{\theta}})$,
$H(\widetilde{\boldsymbol{\theta}},\widehat{\boldsymbol{\theta}})$ and the
test-statistic based on the Hellinger distance, $T_{-0.5}(\overline
{\boldsymbol{p}},\boldsymbol{p}(\widetilde{\boldsymbol{\theta}}%
),\boldsymbol{p}(\widehat{\boldsymbol{\theta}}))$, can be recommended for
intermediate proportions and moderate sample sizes, however for small sample
sizes the likelihood ratio test-statistic still remains being the best
one.

\end{document}